\documentclass[aps,prl,10pt,twocolumn,footinbib,superscriptaddress,amsmath,amsfonts,amssymb,floatfix,notitlepage]{revtex4-2} 

\usepackage{CJK}
\usepackage{physics} 
\usepackage{soul} 
\usepackage[caption=false,subrefformat=parens,labelformat=parens]{subfig}
\usepackage{color} 
\usepackage[table,dvipsnames]{xcolor}
\usepackage[colorlinks,linktocpage,bookmarks=false]{hyperref}
\usepackage[normalem]{ulem}
\usepackage[percent]{overpic}
\usepackage[T1]{fontenc}
\usepackage{graphicx}

\usepackage[legacycolonsymbols]{mathtools}
\usepackage{multirow}
\usepackage{tikz}
\usetikzlibrary{arrows.meta,positioning,shapes.geometric}
\usepackage{dsfont}
\usetikzlibrary{arrows.meta, positioning, shapes.geometric, calc}
\usepackage{xcolor}
\usetikzlibrary{arrows.meta, shapes.geometric}
  
\newcommand\myshade{85}
\definecolor{myrulecolor}{RGB}{150,20,0}
\colorlet{mylinkcolor}{violet}
\colorlet{mycitecolor}{YellowOrange}
\colorlet{myurlcolor}{Aquamarine}
\hypersetup{
	linkcolor  = myurlcolor!\myshade!black,
	citecolor  = mycitecolor!\myshade!black,
	urlcolor   =  myrulecolor!\myshade!black
}

\graphicspath{{FIG/}}
\usepackage{multirow} 
\usepackage{makecell}
\usepackage{booktabs}

\AtBeginDocument{
	\heavyrulewidth=.08em
	\lightrulewidth=.05em
	\cmidrulewidth=.03em
	\belowrulesep=.65ex
	\belowbottomsep=0pt
	\aboverulesep=.4ex
	\abovetopsep=0pt
	\cmidrulesep=\doublerulesep
	\cmidrulekern=.5em
	\defaultaddspace=.5em
}

\renewcommand\[{\begin{equation}}
\renewcommand\]{\end{equation}}

\newcommand{\R}{\mathbb{R}}

\usepackage{amsthm}
\newtheorem{theorem}{Theorem}

\newtheorem{lemma}{Lemma}

\newcommand{\prlsection}[1]{\par\vspace{0.7\baselineskip} \noindent \textbf{#1}\ ---\ }

\makeindex

\begin{document} 
\begin{CJK*}{UTF8}{gbsn}
		\title{Symmetry-Protected Pinch Curves in Classical Spin Liquids}
        \author{Takumi Fukushima}  
        \email{tfukushima@issp.u-tokyo.ac.jp}
        \affiliation{Institute for Solid State Physics, University of Tokyo, Kashiwa, Chiba 277-8581, Japan} 
        \affiliation{RIKEN Center for Emergent Matter Science, Wako, Saitama 351-0198, Japan}
        \author{Han Yan (闫寒)} 
        \email{hanyan@issp.u-tokyo.ac.jp}
        \affiliation{Institute for Solid State Physics, University of Tokyo, Kashiwa, Chiba 277-8581, Japan} 
\date{\today}

\begin{abstract}
Classical spin liquids are correlated paramagnets in which local constraints generate extensive degeneracy and emergent gauge structures, often observable as pinch-point singularities in spin structure factors. Here we introduce pinch-curve spin liquids, in which the pinch singularities form one-dimensional algebraic curves in momentum space. Inversion symmetry protects these curves by reducing the singularity condition to two real algebraic constraints in three dimensions, and the geometry of the pinch locus is algebraically programmable. We identify elementary mechanisms for generating straight and curved pinch loci, construct lattice spin models that realize them, and test the predicted structure factors using Monte Carlo simulations. 
We further show that pinch curves can host an infrared Gauss-law transition: the leading local constraint and the associated anisotropic scaling of the structure factor change, even though the singular locus remains one-dimensional.
\end{abstract}

\maketitle
\end{CJK*}
 
\prlsection{Introduction}
Classical spin liquids are correlated paramagnets whose fluctuations are governed by local constraints rather than conventional order.  These constraints often produce algebraic correlations and singularities in spin structure factors~\cite{MoessnerChalker1998PRL,MoessnerChalker1998,IsakovGregorMoessnerSondhi2004,Henley2005,Henley2010,YanBentonNevidomskyyMoessner2024Detailed,YanBentonMoessnerNevidomskyy2024Typology,FangCanoNevidomskyyYan2024}. The canonical example is classical spin ice on the pyrochlore lattice, where the ice rule realizes an emergent electrostatics, or Coulomb phase.  Its characteristic experimental fingerprint is the pinch-point singularity, which reflects the transverse projector imposed by Gauss's law~\cite{BramwellGingras2001,CastelnovoMoessnerSondhi2008,Fennell2009,Henley2010}.

The landscape of such constrained paramagnets has broadened significantly with the development of generalized Gauss's laws.  Higher-rank and higher-form gauge theories, anisotropic Gauss's laws, and type-II fracton Gauss's laws introduce generalized notions of charge conservation and multipole conservation beyond Maxwell's electrodynamics~\cite{Pretko2017Subdimensional,Pretko2017GeneralizedElectromagnetism,GaiottoKapustinSeibergWillett2015,Chung2025PhysRevB,VijayHaahFu2015,VijayHaahFu2016,NandkishoreHermele2019,SlagleKim2017,YanReuther2022}.
The corresponding spin structure factors exhibit richer singular features, including higher-rank pinch points, scale-variant pinch points, and straight pinch lines~\cite{BentonJaubertYanShannon2016,PremVijayChouPretkoNandkishore2018,YanBentonJaubertShannon2020,BentonMoessner2021,NiggemannIqbalReuther2023,DavierGomezAlbarracinRosalesPujolJaubert2025}.

These singularities are not merely different patterns in momentum space. They diagnose distinct emergent gauge structures and therefore distinct physical consequences, including charge mobility, topological sectors, and possible roles as classical parent states for quantum spin liquids, topological orders, and fracton orders~\cite{HermeleFisherBalents2004,SavaryBalents2012,LeeOnodaBalents2012,GingrasMcClarty2014,Pretko2017Subdimensional,Pretko2017GeneralizedElectromagnetism,PhysRevB.97.235112,PhysRevB.98.035111}, as well as topological bands~\cite{YanPohleShannon2018,YanRomhanyiThomasenShannon2024Berry,Yan2023TopologicallyCriticalBand}. Identifying classical spin liquids with new types of singularities therefore reveals  routes to new emergent gauge theories, quantum phases of matter, and constrained dynamics.

Here we expand this landscape by introducing symmetry-protected \emph{pinch-curve spin liquids}.
These are classical spin liquids described by generalized Gauss's laws, in which the correlation singularities form algebraic curves in momentum space. 
The key mechanism is a symmetry-induced reduction in the number of algebraic constraints that define the pinch-singularity locus. 
In particular, inversion symmetry fixes the parity of the Gauss-law operators and reduces the singularity condition to two real constraints in three-dimensional space. Their common zero set generically forms a one-dimensional curve carrying the pinch singularities.  

We develop the theory as follows. First, we derive the symmetry-protection mechanism using the scalar-charge Gauss's law Hamiltonian. We then show how the polynomial form of the constraint controls the singular geometry, producing straight pinch lines in homogeneous models and genuinely curved pinch loci in  inhomogeneous ones. We then realize representative cases in lattice spin models and verify their structure factors by Monte Carlo simulations. Finally, we show that the observable singular locus can remain one-dimensional even when the leading local constraint changes its differential order.

\prlsection{Spin liquids from generalized Gauss's laws}
We begin from the minimal continuum theory of a generalized Gauss's law for a two-component electric field $\vb*E$,
\begin{equation}
    \vb* D(\vb* \partial) \cdot \vb* E \equiv D_1E_1+D_2E_2=\rho,
    \label{eq:generalized_gauss_law}
\end{equation}
where $E_{1,2}$ are  electric-field components, $D_{1,2}$ are finite-order local differential operators with 
no zeroth-order term and no common factor, and $\rho$ denotes the charge density. The corresponding quadratic Hamiltonian is diagonal in momentum space and has the matrix form
\begin{equation}
\begin{gathered}
    \mathcal H =
    \int d^3k \, \widetilde E_a^\ast(\vb*{k}) \, \mathcal{T}_{ab}(\vb* k) \, \widetilde E_b(\vb*{k}), \\
   \mathcal{T}_{ab}(\vb* k) = D_a(-i \vb*{k})D_b(i \vb*{k}),
\end{gathered}
\label{eq:rank_one_hamiltonian_momentum_space}
\end{equation}
with summation over \(a,b=1,2\).
$D_a(i\vb*{k})$ is the Gauss-law \textit{symbol}, \textit{i.e.}, the Fourier mode of the differential operator. 
Away from the common zeros of $D_1$ and $D_2$, the matrix $\mathcal{T}_{ab}(\vb* k)$ has rank one.
Its eigenvalues are
\begin{equation}
    \omega_B(\vb*{k})=0,
    \qquad
    \omega_T(\vb*{k})=|D_1(i\vb*{k})|^2+|D_2(i\vb*{k})|^2.
    \label{eq:rank_one_eigenvalues}
\end{equation}
The flat band represents fluctuations within the charge-free ground-state manifold, while the dispersive mode measures the energetic cost of charge-sector fluctuations. This dispersive mode becomes gapless on the common zero set $\mathcal Z=\{\vb*{k}:D_1(i\vb*{k})=D_2(i\vb*{k})=0\}$.
 
For the normalized zero-mode tensor structure, the equal-time correlations are obtained as the projection onto the zero-eigenvalue subspace of the  matrix \cite{YanBentonMoessnerNevidomskyy2024Typology,YanBentonNevidomskyyMoessner2024Detailed},
\begin{equation}
\langle \tilde{E}_a (\vb*{k})\tilde{E}_b (-\vb*{k})
\rangle =\delta_{ab}-\frac{D_a (-i\vb*{k})D_b(i\vb*{k})}{\omega_T(\vb*{k})}.
    \label{eq:flat_band_projector}
\end{equation}
For the hard-spin ensemble used below, this projector fixes the singular tensor structure, up to a smooth scalar form factor.
Thus, near the set \(\mathcal Z\), the projector develops a direction-dependent singularity. An isolated common zero can give the familiar pinch point. 

We study spin liquids in which the set $\mathcal Z$ forms a continuous locus of pinch singularities.
We call such states \textit{pinch-curve spin liquids}.
We demonstrate that these spin liquids can be protected by symmetry and exhibit properties absent in conventional pinch-point spin liquids.

\prlsection{Symmetry-protected pinch curves}
For real differential operators, \(D_a(i\vb*{k})\) generally contains both even-order real terms and odd-order imaginary terms. Thus, the two equations \(D_1(i\vb*{k})=D_2(i\vb*{k})=0\) impose four real constraints in three-dimensional space, and extended common zeros are then nongeneric;  without an additional symmetry, only isolated zeros, such as $\vb* k = \vb* 0$, are stable. Therefore, pinch curves require a symmetry that reduces the condition to two real constraints.   

A concrete way to realize this protection is to combine spatial inversion with an internal orthogonal transformation of the two electric-field components.
Given $\vb* D(i\vb*{k})$,
we consider an inversion symmetry of the form
\begin{equation}
    \mathcal I:\qquad
    \vb*r\to-\vb*r,
    \qquad
    \vb* E(\vb*r)\to U\,\vb* E(-\vb*r),
    \label{eq:inversion_internal_mixing}
\end{equation}
where \(U\) is a real orthogonal matrix acting on two-component electric-field space. Since the real-space differential operators have real coefficients, one has $\vb* D(i\vb*{k})^\ast=\vb* D(-i\vb*{k})$.
Invariance under Eq.~\eqref{eq:inversion_internal_mixing} then imposes the reality condition
\begin{equation}
    \vb* D(i\vb*{k}) = U\,\vb* D(i\vb*{k})^\ast.
    \label{eq:general_inversion_reality}
\end{equation}
This is the essential codimension-reducing condition.  Instead of two unrelated complex equations, the two components of the Gauss-law symbol are restricted to a two-dimensional real subspace. Therefore \(D_1(i\vb*{k})=D_2(i\vb*{k})=0\) gives only two independent real constraints.

With this convention, an order-two inversion symmetry requires \(U^2=\mathds{1}\).
By an orthogonal change of basis in the electric-field space, any such \(U\) can be brought to one of three diagonal forms:
\begin{equation}
    U=+\mathds{1},
    \qquad
    U=-\mathds{1},
    \qquad
    U=\mathrm{diag}(1,-1).
    \label{eq:three_inversion_classes}
\end{equation}
Physically, different choices of \(U\) with the same set of eigenvalues can represent distinct symmetry actions. A diagonal \(U\) means that inversion leaves the two electric-field components separately invariant or odd, while an off-diagonal \(U\) means that inversion exchanges/mixes the two components, as happens naturally when \(E_1,E_2\) live on symmetry-related sublattices.  
\begin{table}[t]
\caption{Representative inversion symmetry constraints on Gauss-law symbols. The first three rows are diagonal normal forms; the exchange row is a physically useful representative of the mixed even/odd class. Here \(P_a\) are real polynomials.}
\label{tab:symmetry_constraints}
\begin{ruledtabular}
\renewcommand{\arraystretch}{1.15}
\begin{tabular}{cc}
 Action \(U\)  &
 Constraint on the symbols  \\
\hline
\(+\mathds{1}\) &
\(D_a(i\vb*{k})=P_a(\vb*{k})\), with \(P_a\) even \\[1mm]
\(-\mathds{1}\) &
\(D_a(i\vb*{k})=iP_a(\vb*{k})\), with \(P_a\) odd \\[1mm]
\(\mathrm{diag}(1,-1)\) &
\(D_1=P_1,\ D_2=iP_2\): one even, one odd \\[1mm]
\(\sigma_x\) &
\(D_1(i\vb*{k})=D_2(i\vb*{k})^\ast\): one complex equation
\end{tabular}
\renewcommand{\arraystretch}{1.0}
\end{ruledtabular}
\end{table}
The first three rows of Table~\ref{tab:symmetry_constraints} are the diagonal cases in Eq.~\eqref{eq:three_inversion_classes}.   
In these cases, each \(D_a(i\vb*{k})\) is reduced, up to a phase \(1\) or \(i\), to one real polynomial \(P_a\) of definite parity.
The fourth row gives a representative exchange symmetry, $E_1 \leftrightarrow E_2$.
In this case, $D_a$ is still complex, but the two equations $D_1=0$ and $D_2=0$ are not independent, and they reduce to one complex equation imposing two real constraints.

In the representative examples below we focus on the diagonal cases in Table~\ref{tab:symmetry_constraints}. Then each symbol can be written as
$
    D_a(i\vb*{k})=\chi_a P_a(\vb*{k}), \ \chi_a\in\{1,i\},
$
where \(P_a\) is a real polynomial of definite parity. Since \(\chi_a\neq0\), the symbol zero condition is equivalent to
$
    P_1(\vb*{k})=P_2(\vb*{k})=0.
$
The zero set is therefore an algebraic curve defined by two real equations in three dimensions.  At a regular real solution, where \(\nabla P_1\) and \(\nabla P_2\) are linearly independent, the implicit function theorem gives a smooth one-dimensional branch. Symmetry-preserving perturbations can deform such a branch, but cannot remove it locally without a singular event. This is the sense in which the pinch curve is protected.  
The spin correlations are governed by the constrained projector in Eq.~\eqref{eq:flat_band_projector}, which becomes directionally singular along this common-zero locus.

Each point on a pinch curve also carries a local emergent Gauss's law. For \(\vb*{k}=\vb*{k}_0+\vb*q\), the leading terms of $\vb*q$ in the Taylor expansion of \(D_a(i\vb*{k}_0+i\vb*q)\) define the infrared Gauss's law near \(\vb*{k}_0\)~\cite{YanBentonMoessnerNevidomskyy2024Typology,YanBentonNevidomskyyMoessner2024Detailed}. This local Gauss's law can vary along the curve, and its differential order can change at some points. 
In this regard, the pinch-curve spin liquid, although described by a global Gauss's law, can be viewed as a continuous family of infrared spin liquids along the curve. 
This is a sharp distinction from canonical pinch-point spin liquids. We give an explicit example below.

\prlsection{Algebraic design of elementary pinch curves}
For the diagonal symmetry cases used below, the common-zero condition reduces to real parity-fixed polynomials \(P_a\). The symmetry argument then becomes an algebraic design and classification principle: $\mathcal Z=\{\vb*{k}\in\mathbb R^3:P_1(\vb*{k})=P_2(\vb*{k})=0\}$.
     
The polynomial form of the generalized Gauss's law fixes the geometry of the pinch locus. If both constraints are homogeneous, every nonzero solution generates the full line \(\lambda\vb*{k}\); homogeneous Gauss's laws therefore produce straight pinch lines in the continuum theory.
Genuinely curved pinch loci in the continuum require inhomogeneous constraints, \textit{i.e.}, terms of different orders within the same $P_a$.

\begin{figure}[t]
    \centering
    \includegraphics[width=\linewidth]{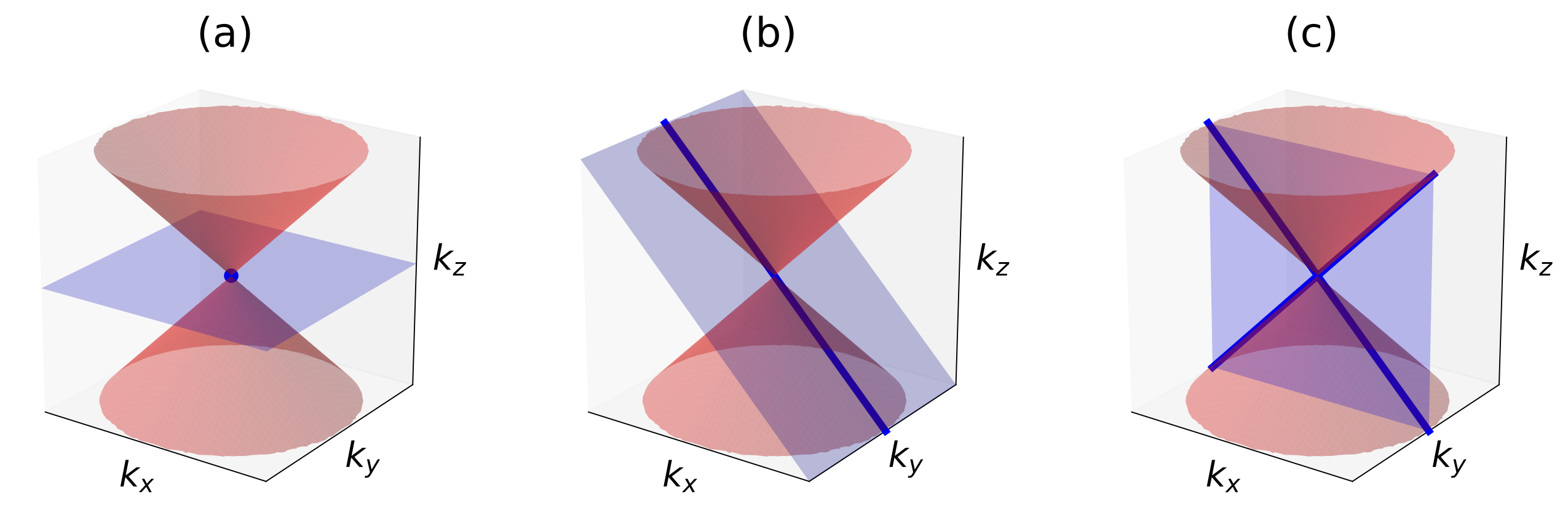}
    \caption[Elementary linear-quadratic intersections.]{Linear-quadratic intersections. A plane \(P_1(\vb*{k})=0\) intersects a quadratic cone \(P_2(\vb*{k})=0\). Dark blue denotes the resulting straight pinch lines. The three cases are (a) \(\eta<0\), origin only; (b) \(\eta=0\), one tangent line; and (c) \(\eta>0\), two distinct lines.}
    \label{fig:cone_plane}
\end{figure}

\begin{figure}[t!]
\centering
\includegraphics[width=\linewidth]{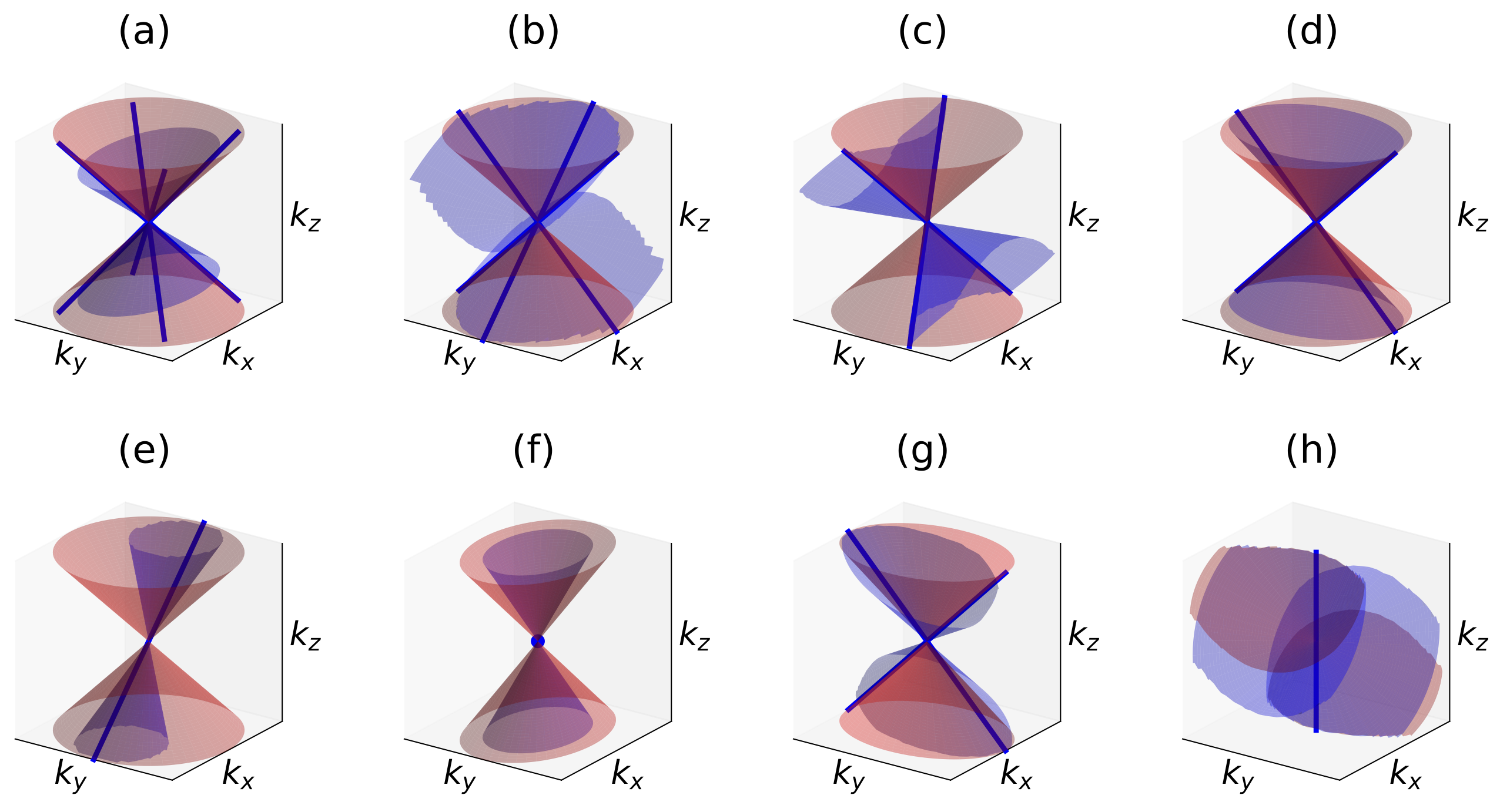} 
\caption[Representative intersection types for two quadratic cones.]{Quadratic-quadratic intersections. Two quadratic cones \(P_1(\vb*{k})=0\) and \(P_2(\vb*{k})=0\) intersect along straight pinch lines, shown in dark blue.}
\label{fig:two_cones}
\end{figure}

\emph{Homogeneous examples.}---
The simplest nontrivial homogeneous mechanism is the linear-linear case: each \(P_a\) contains only linear terms and defines a plane, and the pinch locus is their straight-line intersection. 
The second simplest mechanism is the linear-quadratic case, \(P_1(\vb*{k})=\vb*n^{\mathsf T}\vb*{k}=0\) and \(P_2(\vb*{k})=\vb*{k}^{\mathsf T}\vb*A\vb*{k}=0\).   For a nondegenerate indefinite cone, choosing the sign of \(\vb*A\) so that it has two positive eigenvalues and one negative eigenvalue, the resulting intersection is controlled by
\(\eta=\vb*n^{\mathsf T}\vb*A^{-1}\vb*n\): \(\eta<0\), \(\eta=0\), and \(\eta>0\) correspond, respectively, to the origin only, one tangent line, and two distinct lines. 

The quadratic-quadratic mechanism gives more complex straight-line geometries.  
For \(P_1(\vb*{k})=\vb*{k}^{\mathsf T}\vb*A\vb*{k}\) and \(P_2(\vb*{k})=\vb*{k}^{\mathsf T}\vb*B\vb*{k}\), the intersection of two cones is organized by the \textit{pencil} \(\vb*M(\lambda)=\vb*A-\lambda\vb*B\). 
Simple roots of \(\det\vb* M(\lambda)=0\) reduce the problem to linear-quadratic tests of the type above; the remaining triple-root cases, together with a classification recipe for the real set-theoretic pinch-locus support, are presented in the Supplemental Material (SM).  Representative homogeneous pinch-line geometries are shown in Figs.~\ref{fig:cone_plane} and \ref{fig:two_cones}.

\begin{figure}[t]
    \centering
    \includegraphics[width=\linewidth]{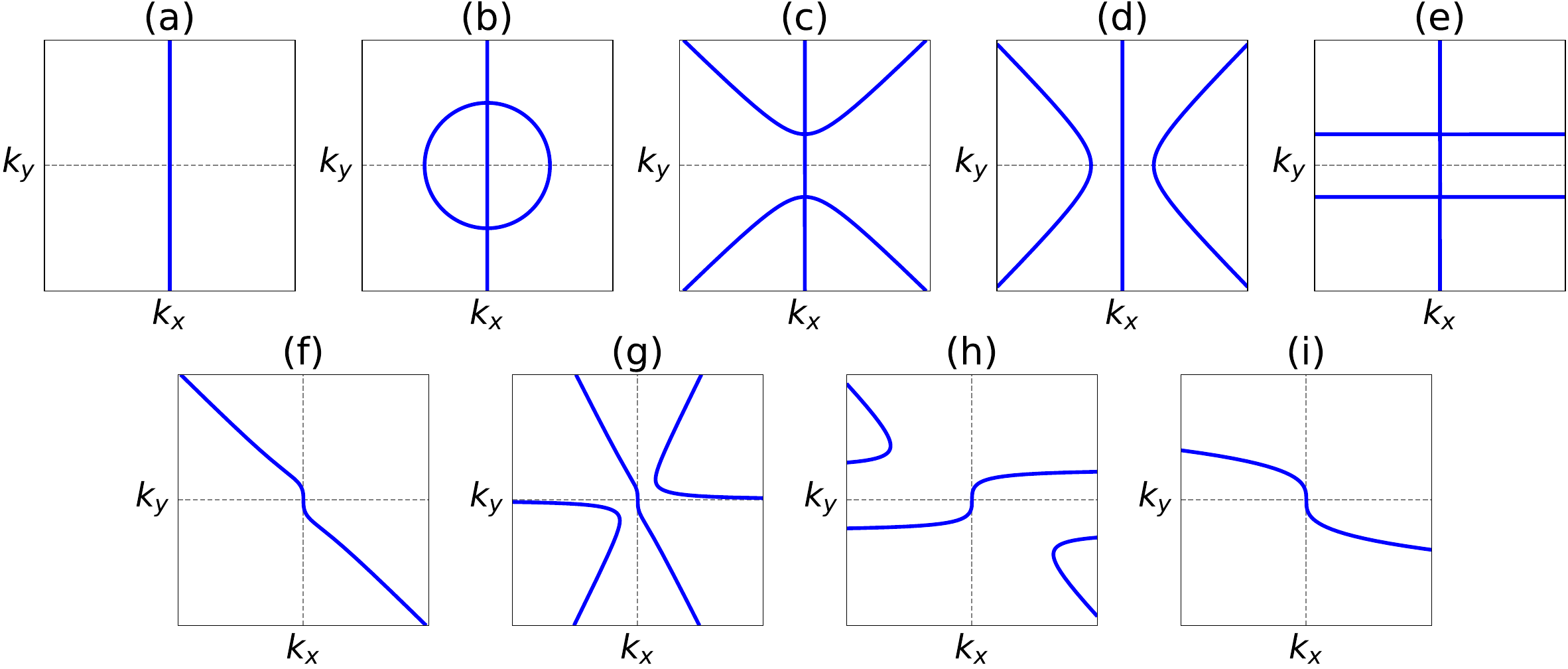}
    \caption { Representative inversion-symmetric planar cubic curves arising from linear-(linear-plus-cubic) intersections.  Reducible cases in panels (a)-(e) contain straight pinch lines, while irreducible cases in panels (f)-(i) give genuine curved pinch loci.}
    \label{fig:cubic_curves}
\end{figure}

\emph{Minimal curved example.}---
Curvature first appears when scale invariance is broken. The minimal odd-parity case consists of a linear constraint and a linear-plus-cubic constraint. After an invertible linear transformation of coordinates, one may set the linear constraint to \(P_1({\vb*k}) =k_z=0\).
Restricting the second constraint to this plane gives a planar cubic of the form $P_2(k_x,k_y) = c_3(k_x,k_y)+\ell_1(k_x,k_y)=0$, where $c_3$ is a homogeneous cubic polynomial and $\ell_1$ is a homogeneous linear polynomial.
A nonzero linear term makes the origin a regular point, so the curve passes smoothly through \(\vb*{k}=0\).
Choosing in-plane coordinates such that $\ell_1=k_x$, one obtains the normal form \begin{equation} P_2(k_x,k_y) = k_x\left(1+a k_x^2+b k_xk_y+c k_y^2\right) +d k_y^3 . \label{eq:linear_cubic_normal_form_main} \end{equation} This form separates the reducible and irreducible cases: $d=0$ gives a reducible curve containing the straight-line component $k_x=0$, while $d\neq0$ gives an irreducible curve with no finite singular point. In the irreducible case, every nonzero point has $k_x\neq0$, and the large-$|\vb*{k}|$ branch structure is governed by $f(t)=c_3(1,t)=a+bt+ct^2+dt^3$,  $t=k_y/k_x$. The real roots of $f(t)$ give the possible asymptotic directions, while their multiplicities determine whether the corresponding branch gives an ordinary asymptote, a pair of asymptotes, or only an asymptotic direction. Thus irreducible constraints provide the minimal genuinely curved pinch loci, whereas reducible constraints contain straight-line pieces.
The classification of the real set-theoretic pinch-locus support, under the stated non-degeneracy assumptions, is given in the SM. Figure~\ref{fig:cubic_curves} shows representative curve types.

These elementary examples capture the main design principle. Homogeneous Gauss-law symbols impose scale invariance and therefore produce straight pinch lines. Inhomogeneous symbols break this scale invariance without violating the symmetry protection, allowing the singular locus to bend into a symmetry-protected curve.  The Taylor expansions of the two constraint polynomials determine the local form of the curve near \(\vb*{k}_0\), the corresponding infrared Gauss's law, and the anisotropic scaling of the associated structure factor.
These results establish a close connection between algebraic geometry and the physics of pinch curves, complementing broader classification and topological approaches to classical spin liquids~\cite{YanBentonNevidomskyyMoessner2024Detailed,YanBentonMoessnerNevidomskyy2024Typology,FangCanoNevidomskyyYan2024,RueeggMorrisYanSlager2025Euler,LozanoGomezBentonGingrasYan2025Atlas}.   

\prlsection{Microscopic lattice realizations and scattering signatures}
We now realize these algebraic mechanisms in explicit lattice spin models and show that the continuum predictions are reproduced near the expansion point. We place two spins in each unit cell of a cubic lattice and identify \(E_{a,\vb*r}=S^z_{a,\vb*r}\) (\(a=1,2\)).  The Hamiltonian takes the local constraint form
\begin{equation}
    H = J\sum_{\vb*r}Q_{\vb*r}^2,
    \qquad
    Q_{\vb*r}=\Delta_1S^z_{1,\vb*r}+\Delta_2S^z_{2,\vb*r},
    \label{eq:lattice_rank_one_hamiltonian}
\end{equation}
where \(\Delta_a\) are discrete lattice difference operators chosen to reproduce the desired Gauss-law symbols $D_a$ in the continuum limit. We denote their Fourier symbols by \(\Delta_a(i\vb*{k})\), which play the same role as the continuum symbols \(D_a(i\vb*{k})\).

We set \(J=1\) and use the shorthand \(s_\mu\equiv\sin k_\mu\), with $\mu=x,y,z$. The three representative choices are 
\begin{align}
\text{I:}\quad & \Delta_1=-4s_x^2+4s_y^2-4s_z^2,\quad   \Delta_2=2is_z, \label{eqn.model1}\\[1mm]
\text{II:}\quad & \Delta_1=s_x^2+2s_y^2-s_z^2,\quad   \Delta_2=2s_x^2+s_y^2-s_z^2, \label{eqn.model2} \\[1mm]
\text{III:}\quad & \Delta_1=8i(s_y-s_x^3),\quad   \Delta_2=2is_z . \label{eqn.model3}
\end{align}
Their zero sets realize, respectively, the three cases discussed above: a plane-cone intersection, an intersection of two quadratic cones, and the curved locus. The corresponding real-space Hamiltonians follow directly by inverse Fourier transforming the finite-difference symbols.

\begin{figure}[t]
\centering
\subfloat[]{\includegraphics[height=0.33\linewidth]{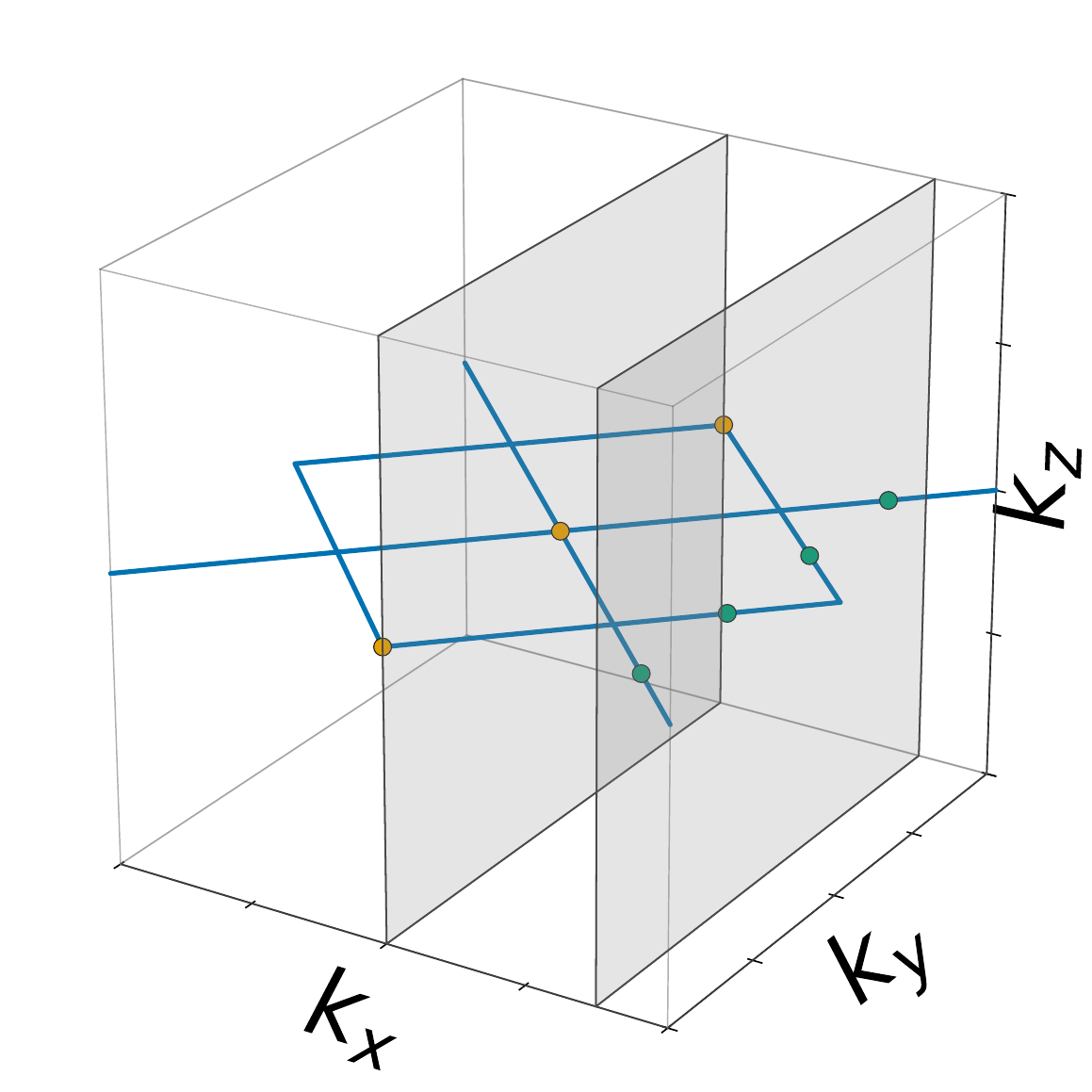}} 
\subfloat[]{\includegraphics[height=0.33\linewidth]{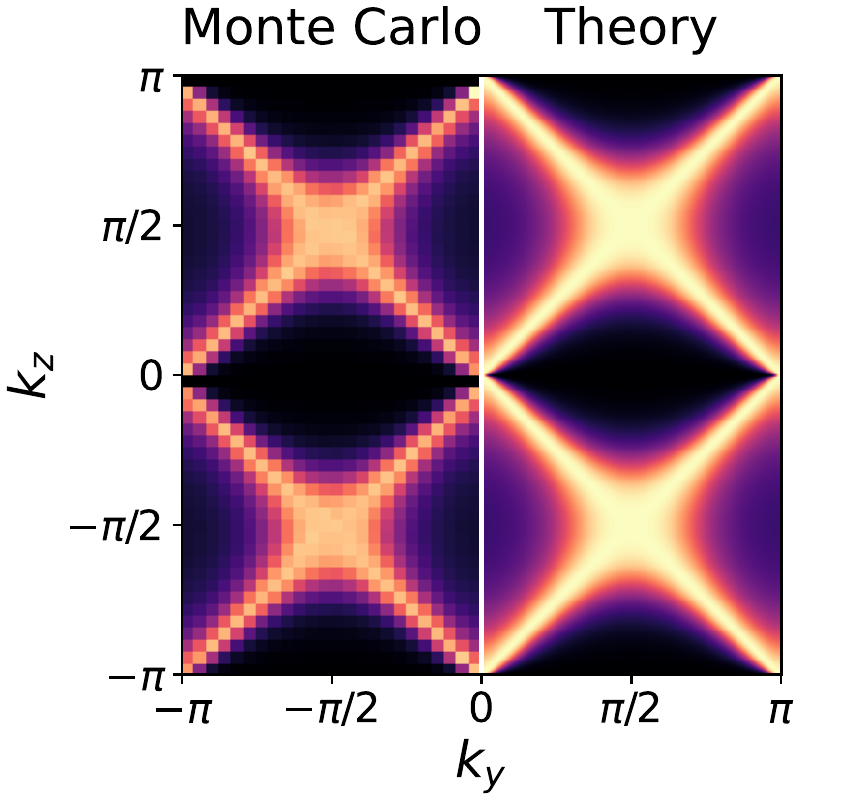}}
\subfloat[]{\includegraphics[height=0.33\linewidth]{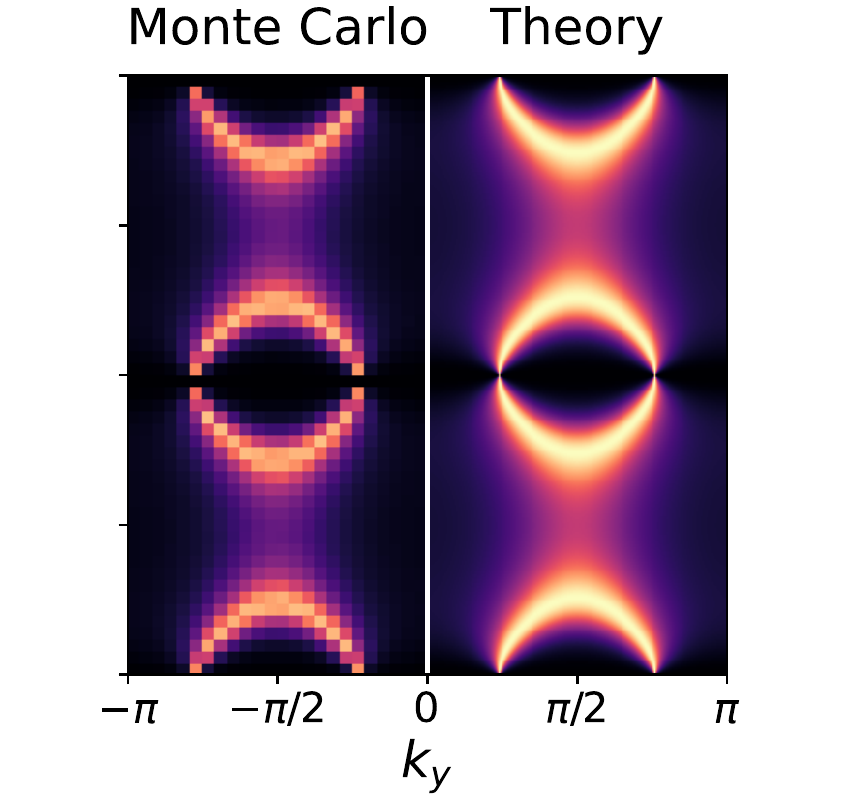}}\\
\subfloat[]{\includegraphics[height=0.33\linewidth]{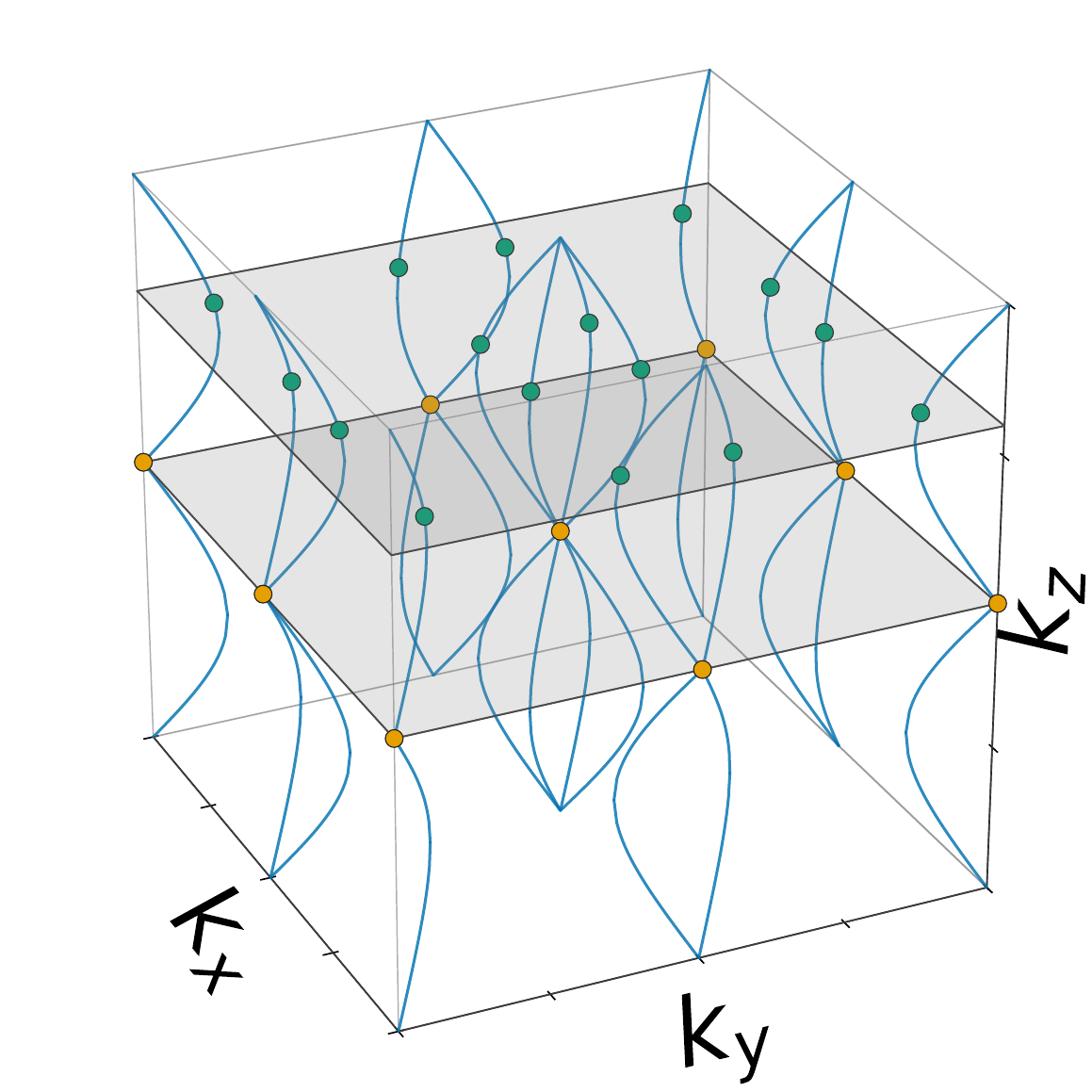}} 
\subfloat[]{\includegraphics[height=0.33\linewidth]{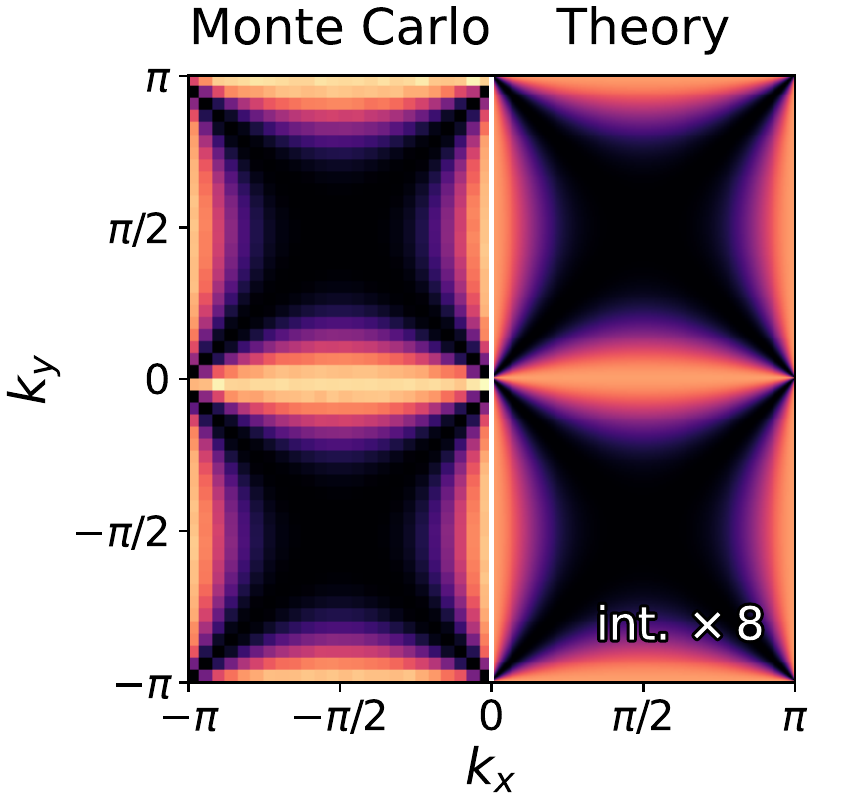}} 
\subfloat[]{\includegraphics[height=0.33\linewidth]{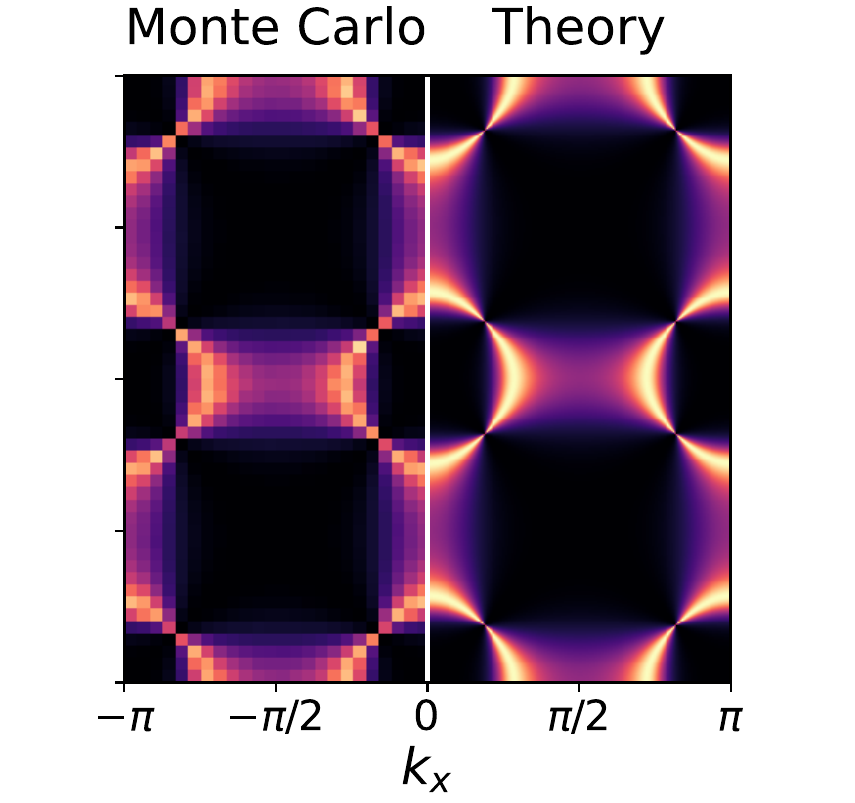}}\\
\subfloat[]{\includegraphics[height=0.33\linewidth]{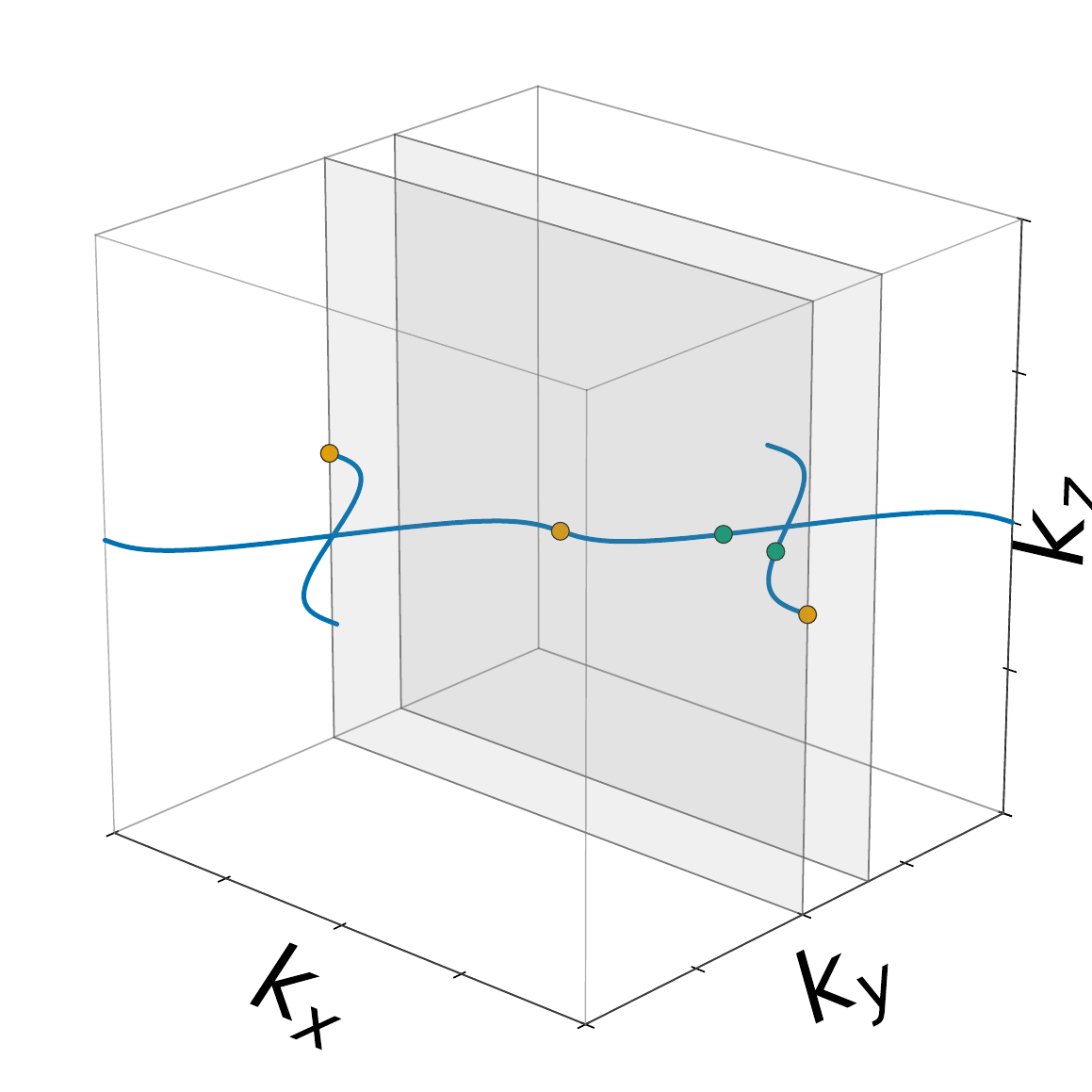}} 
\subfloat[]{\includegraphics[height=0.33\linewidth]{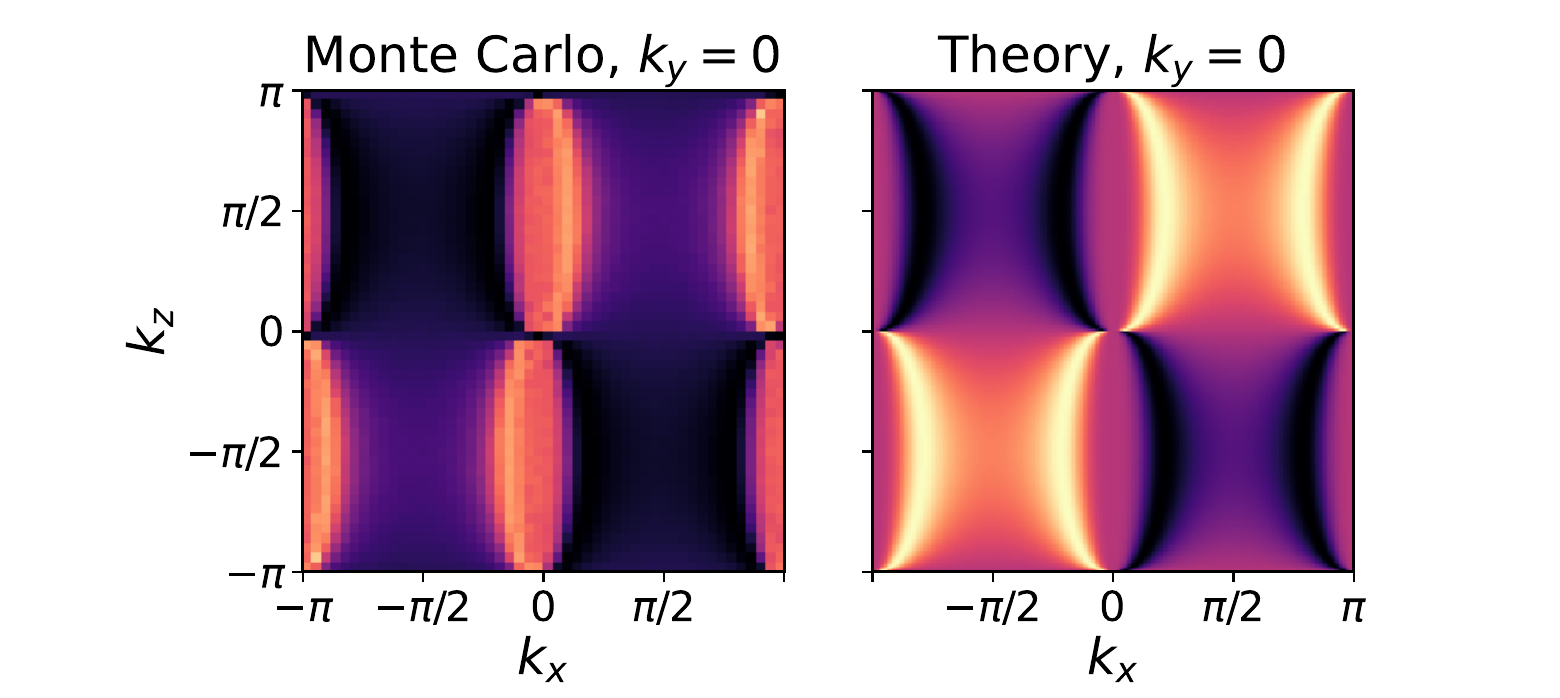}} \\
\subfloat[]{\includegraphics[height=0.33\linewidth]{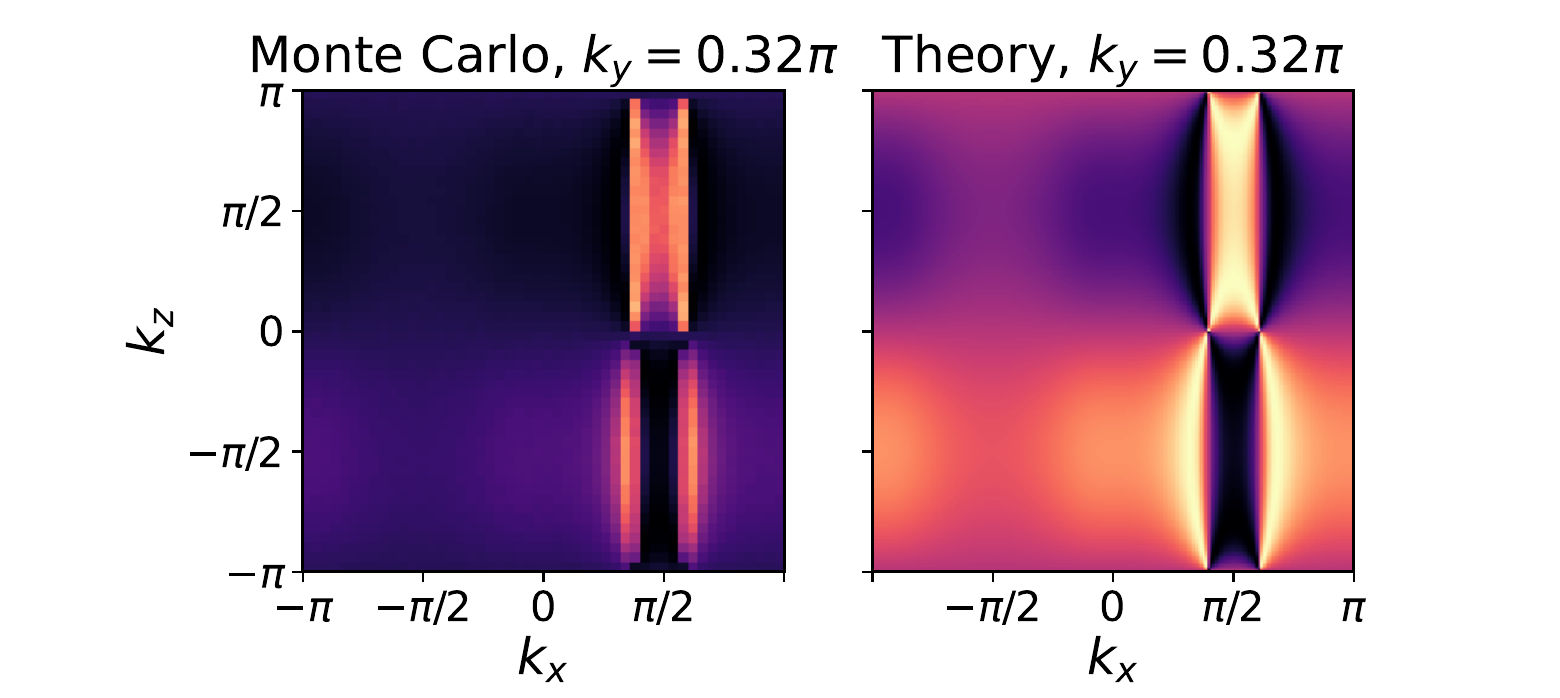}} 
\subfloat{\includegraphics[height=0.33\linewidth]{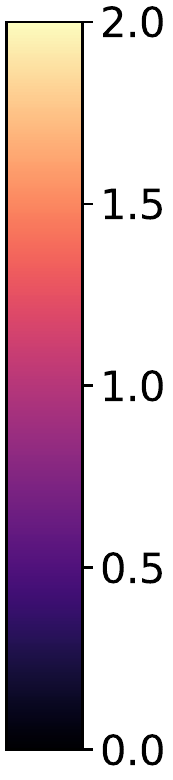}} 
\caption{Structure factors for the three lattice models. (a--c) plane-cone model of Eq.~\eqref{eqn.model1}. 
Panel (a) shows the pinch curve in \(C_{11}(\vb*{k})\) as blue lines in reciprocal space. The gray cross-sections are the planes for the structure factor in panels (b,c). The dots mark the locations of the pinch curve in the respective cross-sections. (d--f) Double-cone model of Eq.~\eqref{eqn.model2}, with $C_{\Sigma}(\vb*{k})$ plotted. (g--i) Inhomogeneous pinch-curve model of Eq.~\eqref{eqn.model3}, with $C_{\Sigma}(\vb*{k})$ plotted.}
\label{fig:three.models}
\end{figure}
 
We measure either the spin-correlation channel \(C_{ab}(\vb*{k})=\langle S^z_a(\vb*{k})S^z_b(-\vb*{k})\rangle\) or the symmetric channel \(C_{\Sigma}(\vb*{k})=\langle S^z_\Sigma(\vb*{k})S^z_\Sigma(-\vb*{k})\rangle\), where \(S^z_\Sigma=S^z_1+S^z_2\).
Monte Carlo simulations are performed on \(L=50\) periodic lattices using a worm update that samples the constrained low-energy manifold.
As shown in Fig.~\ref{fig:three.models}(a--i), the simulated singular loci and structure-factor patterns agree with the zero-mode prediction for all three models.  In Model I  [Fig.~\ref{fig:three.models}(a--c)], different momentum cuts intersect the plane-cone zero set at different points, producing the moving pinch pattern of a straight pinch line. 
In Model II [Fig.~\ref{fig:three.models}(d--f)], the quadratic-quadratic intersection produces four symmetry-related branches. The continuum straight lines become curved in the Brillouin zone because \(k_\mu\) is replaced by \(\sin k_\mu\).

In Model III [Fig.~\ref{fig:three.models}(g--i)], the singular feature follows the curved lattice zero set \(s_y=s_x^3,\ s_z=0\). This demonstrates that an inhomogeneous constraint in the Gauss-law symbol is directly visible in the scattering intensity. The agreement confirms that the universal pinch singularity is controlled by the algebraic zero set of the Gauss-law symbols.
In particular, the inhomogeneous model illustrates the distinctive local structure. Around $\vb* k = \vb* 0$, \(\Delta_1\simeq8i(k_y-k_x^3)\) and \(\Delta_2\simeq2ik_z\), giving
\begin{equation}
    C_\Sigma(\vb*{k})
    \simeq
    \frac{\left[k_z-4(k_y-k_x^3)\right]^2}{k_z^2+16(k_y-k_x^3)^2}.
    \label{eq:non_scale_invariant_pinch}
\end{equation}
Under ordinary isotropic scaling about \(\vb* k=\vb*0\), the cubic term in \(k_y-k_x^3\) is subleading to the linear term \(k_y\). Thus the leading infrared projector at the origin is the regular first-order form controlled by the transverse pair \((k_y,k_z)\). Nevertheless, the cubic term controls the bending of the pinch locus and becomes visible in tangent-sensitive cuts or under the anisotropic scaling \(k_x\sim\lambda\), \(k_y\sim k_z\sim\lambda^3\). Away from points where the gradients of the two lattice constraints become dependent, the pinch curve has the ordinary regular pinch-line form. The tuned model below gives a sharper example in which the leading local Gauss's law itself changes differential order.

\prlsection{Gauss-law transition on a pinch curve}
Pinch curves also host a new type of Gauss-law transition. We demonstrate this with a concrete model in which the singular locus remains one-dimensional but the leading Gauss's law changes its differential order.  This differs from more familiar scenarios in which a transition is driven by the creation, annihilation, splitting, or merging of pinch points. Here, the change is tied instead to a singular change in the local geometry of a pinch curve.

Consider the following one-parameter model family:
\begin{equation}
    \Delta_1(\vb*{k};\delta) = i(\sin^3k_y+\delta\sin k_x),
    \quad
    \Delta_2(\vb*{k}) = i \sin k_z.
    \label{eq:gauss_law_transition_lattice_family}
\end{equation}
Near the origin, the zero set remains one-dimensional for both \(\delta\neq0\) and \(\delta=0\). However, the leading local constraint changes.  Expanding around \(\vb*{k}=\vb0\), one obtains
\begin{equation}
    \Delta_1(\vb*q;\delta) = i(q_y^3+\delta q_x+\cdots),
    \qquad
    \Delta_2(\vb*q) = i q_z+\cdots.
    \label{eq:transition_local_symbols}
\end{equation}
For any fixed \(\delta\neq0\), the linear term \(\delta q_x\) dominates the first constraint at the longest wavelengths.  The leading Gauss's law is therefore equivalent to a first-order conventional one involving \(q_x\) and \(q_z\).  At the tuned point \(\delta=0\), this linear term vanishes, and the leading constraint instead involves \(q_y^3\) and \(q_z\).  In this sense the local infrared constraint changes as
\begin{equation}
    (q_x,q_z)
    \quad\longrightarrow\quad
    (q_y^3,q_z) .
    \label{eq:transition_leading_data_change}
\end{equation}
The visible real locus does not undergo a corresponding dimensional change: for \(\delta\neq0\) it is locally $q_z=0$, $q_x=-\delta^{-1}q_y^3+\cdots$,
whereas at \(\delta=0\) it becomes $q_y=0$, $q_z=0$. 
Both are one-dimensional, but their leading transverse constraints are different.
Figure~\ref{fig:gauss_law_transition} visualizes this transition in both the singular curve [panels (a--c)] and the corresponding structure factor [panels (d--f)].

\begin{figure}[t]
    \centering 
\subfloat[]{\includegraphics[width=0.33\linewidth]{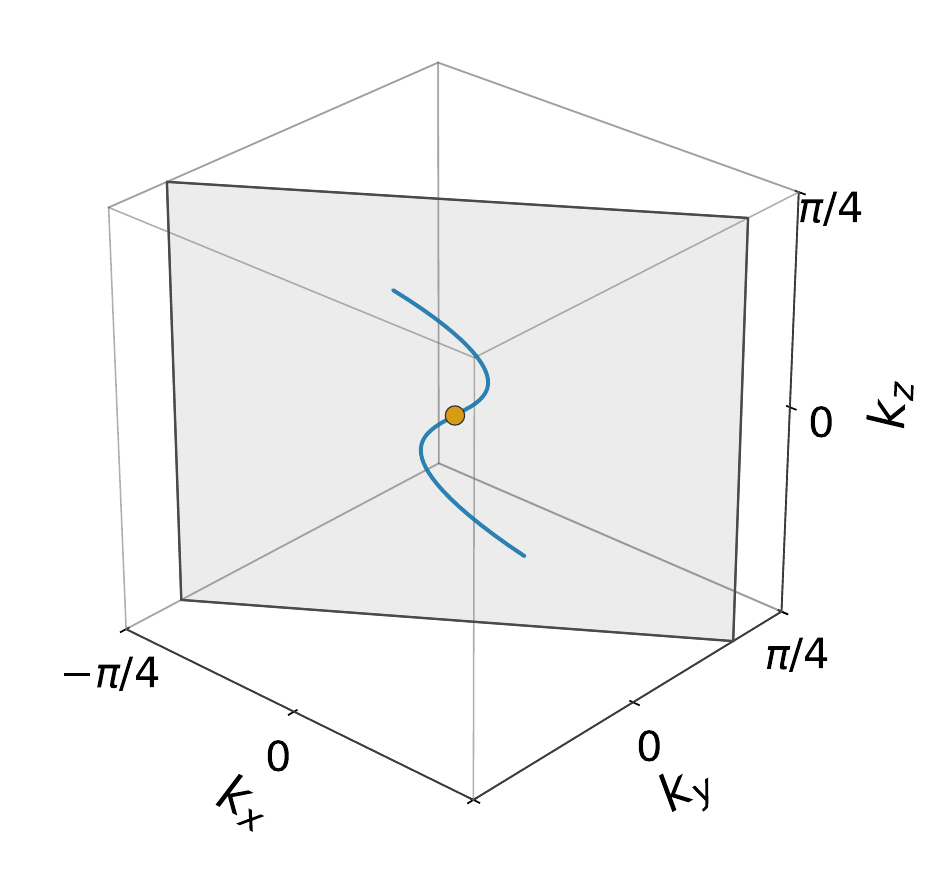}} 
\subfloat[]{\includegraphics[width=0.33\linewidth]{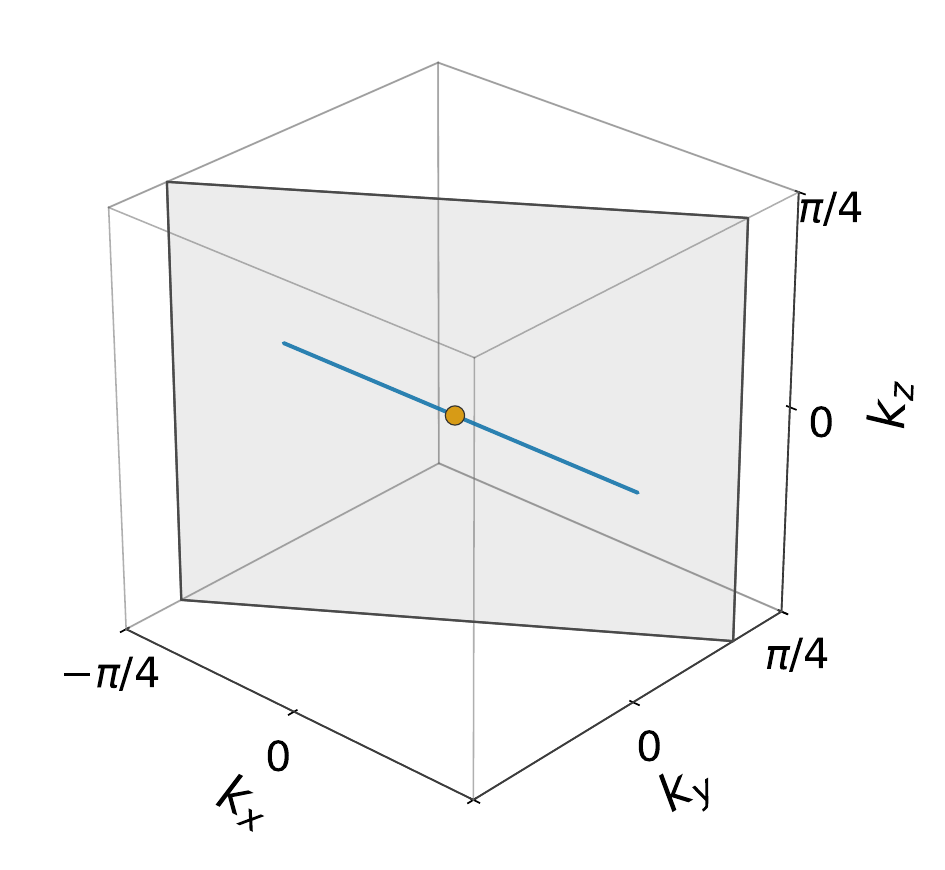}}
\subfloat[]{\includegraphics[width=0.33\linewidth]{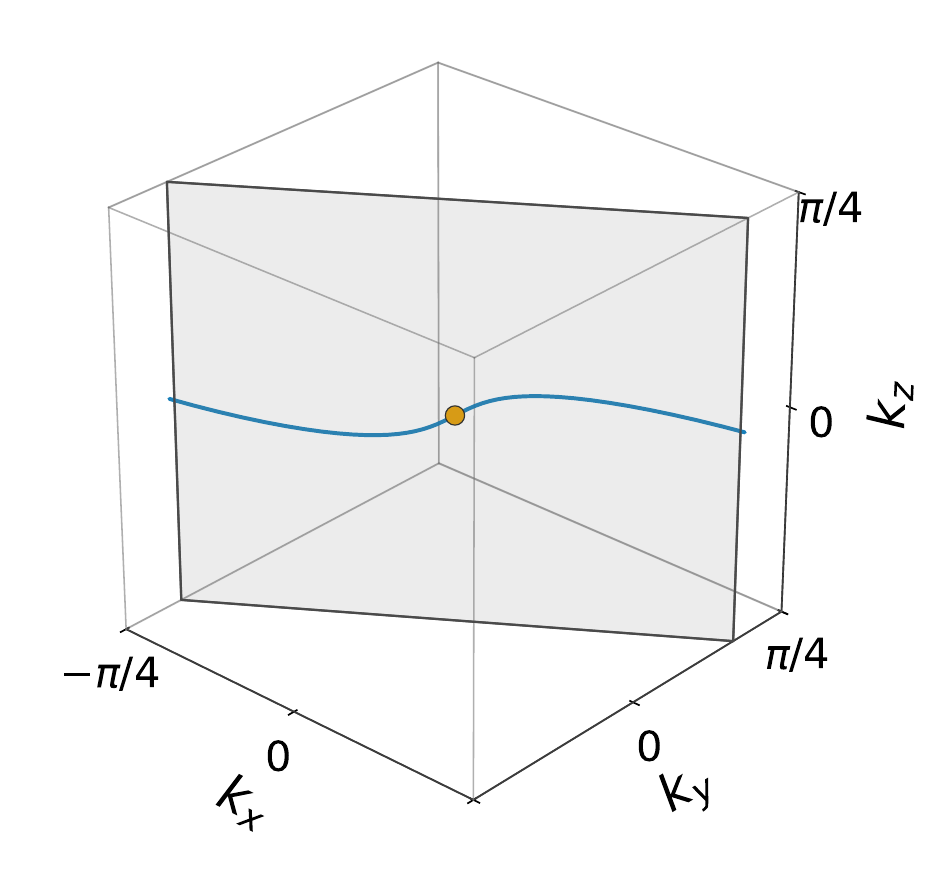}}\\
\subfloat[]{\includegraphics[width=0.33\linewidth]{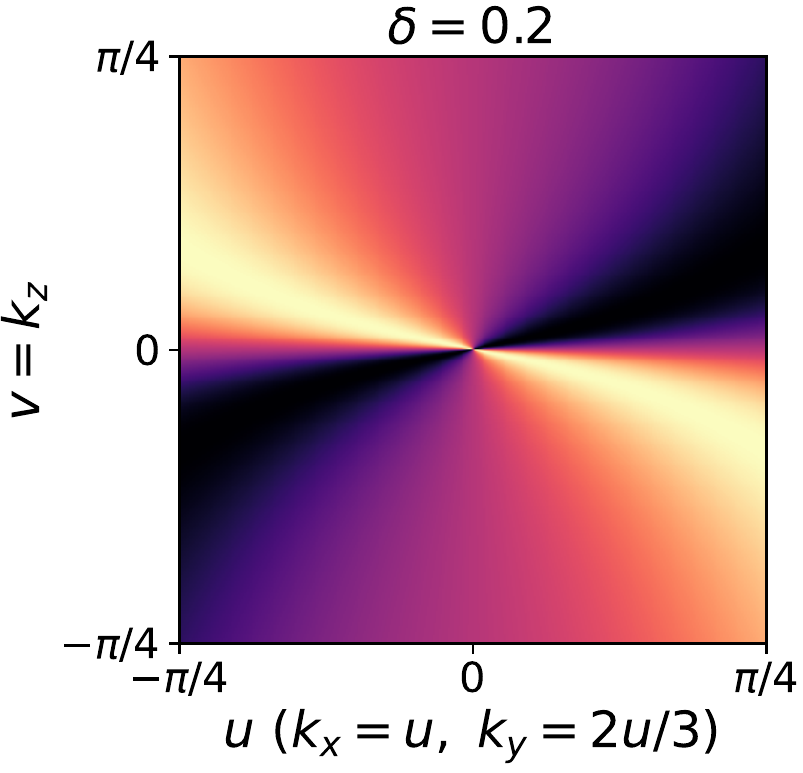}} 
\subfloat[]{\includegraphics[width=0.33\linewidth]{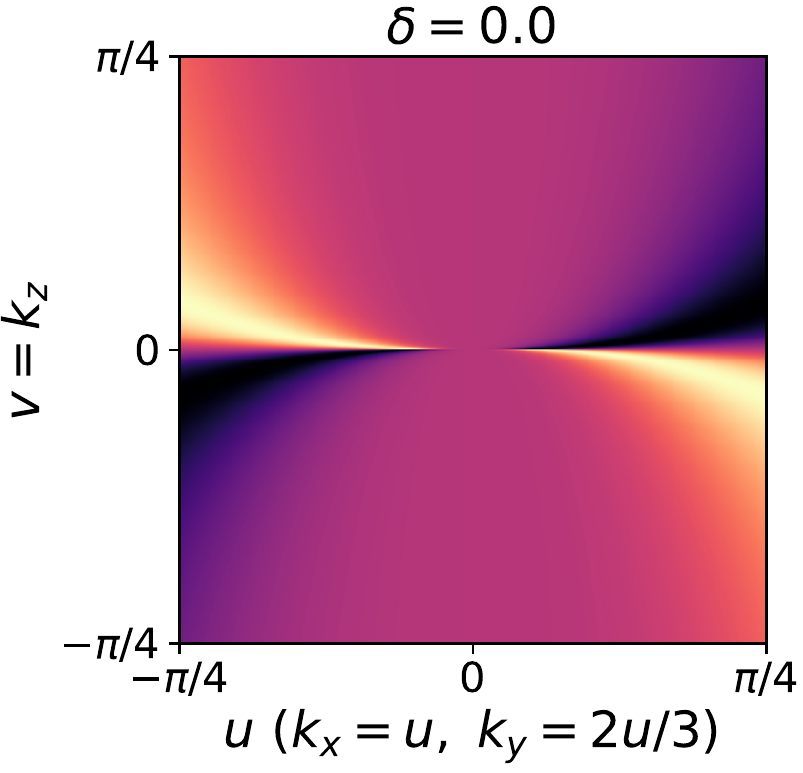}} 
\subfloat[]{\includegraphics[width=0.33\linewidth]{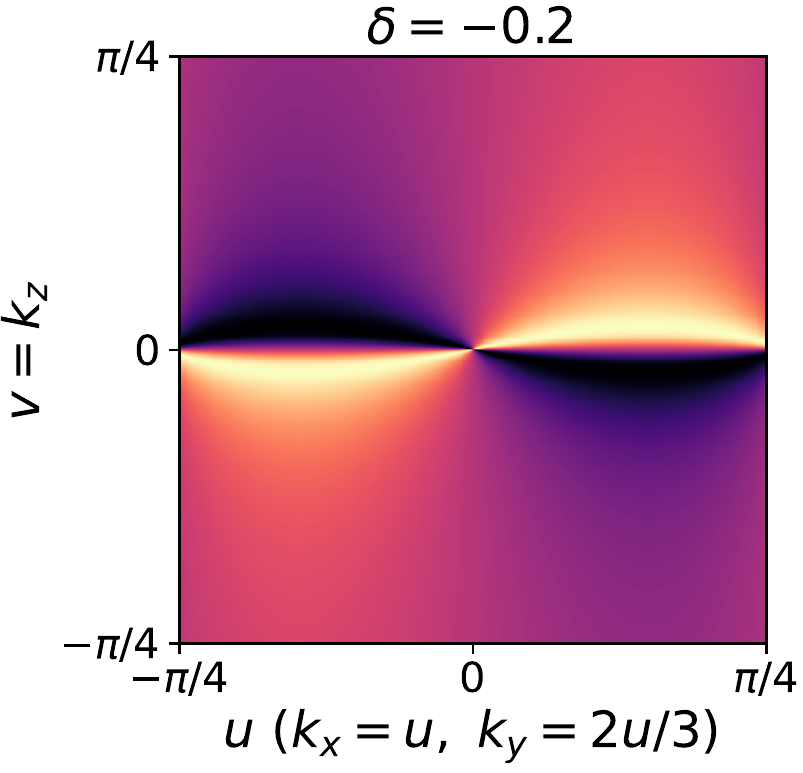}} 
\caption{Gauss-law transition for \(\Delta_1=i(\sin^3k_y+\delta\sin k_x)\), \(\Delta_2= i \sin k_z\). Panels (a,b,c) show the singular curve for \(\delta=0.2,0,-0.2\) and the momentum cuts; panels (d,e,f) show the corresponding structure factors. The leading local constraint changes from first order for \(\delta\neq0\) to third order at \(\delta=0\).}
\label{fig:gauss_law_transition}
\end{figure}

The same change appears directly in the structure factor. The zero-mode prediction for the symmetric channel has the local form
\begin{equation}
    C_\Sigma(\vb*q;\delta)
    \simeq
    \frac{\left[q_z-(q_y^3+
    \delta q_x)\right]^2}{q_z^2+(q_y^3+\delta q_x)^2}.
    \label{eq:transition_structure_factor_general}
\end{equation}
For any fixed nonzero \(\delta\), the leading singularity is controlled by the first-order pair \((\delta q_x,q_z)\). At \(\delta=0\), the structure factor instead exhibits anisotropic scaling \(q_z\sim q_y^3\). Thus the tuning is singular in the infrared: the visible zero set remains a curve, but the leading local Gauss's law changes its differential order.

\prlsection{Discussion}
Two-component scalar-charge Gauss's laws in three dimensions provide a minimal setting for generalized gauge theory for  constrained paramagnets. Nevertheless, they exhibit a broad range of phenomena revealed in this work. 

A richer landscape of classical spin liquids, their quantum counterparts, and their transitions remains to be explored within this family and beyond.
First, a complete classification of global Gauss's laws and their local infrared equivalence classes near pinch loci remains an open problem.
Such classification should clarify when different local constraints are equivalent in the infrared and what physical consequences follow, particularly regarding charge mobility and topological sectors.
Second, the connection between algebraic geometry and spin-liquid correlations deserves further development, building on recent classification frameworks for classical spin liquids~\cite{YanBentonNevidomskyyMoessner2024Detailed,YanBentonMoessnerNevidomskyy2024Typology,FangCanoNevidomskyyYan2024,RueeggMorrisYanSlager2025Euler,LozanoGomezBentonGingrasYan2025Atlas}. 
Algebraic notions such as reducibility, multiplicity, singularity, and global winding on the Brillouin torus may provide useful invariants of constrained spin-liquid structures.  

Finally, it is important to identify experimental platforms in which pinch curve spin liquids can be realized or approximated. In particular, frustrated magnets with competing spin-liquid regimes~\cite{YanBentonJaubertShannon2017PhaseCompetition,TaillefumierBentonYanJaubertShannon2017,LozanoGomezYan2025Fragmented} and cold-atom systems with highly tunable interactions~\cite{SemeghiniLevineKeeslingEtAl2021Science,TianSamajdarGadway2025PRL,HalimehAidelsburgerGrusdtHaukeYang2025NaturePhysics} are promising directions.

\section{Acknowledgments}
\begin{acknowledgments} 
T.F. was supported by RIKEN Junior Research Associate Program and the Sasakawa Scientific Research Grant from the Japan Science Society.
H.Y. acknowledges the Grant-in-Aid for Research Activity Start-up from the Japan Society
for the Promotion of Science (Grant No. 24K22856).
\end{acknowledgments}
 
\bibliography{pinch_curve_ref.bib}

\clearpage
\setcounter{equation}{0}
\setcounter{figure}{0}
\setcounter{table}{0}
\makeatletter
\renewcommand{\theequation}{S\arabic{equation}}
\renewcommand{\thefigure}{S\arabic{figure}}
\renewcommand{\thetable}{S\arabic{table}}
\renewcommand{\bibnumfmt}[1]{[#1]}
\renewcommand{\citenumfont}[1]{#1}

\onecolumngrid
 
\begin{center}
	\Large{\textbf{Supplemental Material for ``Symmetry-Protected Pinch Curves in Classical Spin Liquids''}}
\end{center} 

\section{Proofs of the algebraic classifications}

We prove the classification statements used in the main text. 
The classification below determines the real set-theoretic support of the pinch locus under the stated assumptions. In the homogeneous cases this support is a finite union of real lines, while in the linear-(linear-plus-cubic) case it may include genuine affine curve components. It does not keep track of complex projective data, intersection multiplicities, or scheme-theoretic structure.
We analyze the three cases considered there: linear-quadratic, quadratic-quadratic, and linear-(linear-plus-cubic) cases.
For each case, we show explicitly how the algebraic conditions determine the resulting intersection geometry in three-dimensional space.

Before considering individual cases, we first note a general consequence of homogeneity. 
Suppose that the two constraints are homogeneous polynomials,
\[
    P_1(\vb*{k})=0, \qquad P_2(\vb*{k})=0.
\]
Here $P_1(\vb*{k})$ and $P_2(\vb*{k})$ are homogeneous of degree $n$ and $m$, respectively.
If a nonzero vector $\vb*{k}_0$ satisfies $P_1(\vb*{k}_0)=0$ and $P_2(\vb*{k}_0)=0$, then every rescaled vector $\lambda\vb*{k}_0$ with $\lambda\in\R$ also satisfies 
\[
    P_1(\lambda\vb*{k}_0)=\lambda^n P_1(\vb*{k}_0)=0, 
    \qquad 
    P_2(\lambda\vb*{k}_0)=\lambda^m P_2(\vb*{k}_0)=0.
\]
Therefore, each nonzero solution represents a straight line through the origin in three-dimensional space.
Under the coprimality assumption for the two polynomials, the homogeneous pairs can produce only a finite number of straight lines. 
Genuine curved affine loci require inhomogeneous constraints, such as the linear-(linear-plus-cubic) case.

\subsection{The linear-quadratic case}

We prove the classification statement for the linear-quadratic case, summarized in Table~\ref{tab:LQ_classification}. 
The two equations are
\[
    P_1(\vb*{k})=\vb*{n}^{\mathsf T}\vb*{k}=0, 
    \qquad 
    P_2(\vb*{k})=\vb*{k}^{\mathsf T}\vb*{A}\vb*{k}=0,
\]
where $\vb*{n}$ is the normal vector of the plane $P_1(\vb*{k})=0$, and $\vb*{A}$ is a real, symmetric, nondegenerate, and indefinite $3\times3$ matrix. 
We work in the principal-axis basis of $P_2(\vb*{k})$, where
\[
    \vb*{A}=\mathrm{diag}(a_1,a_2,a_3), 
    \qquad 
    a_1,a_2>0, 
    \qquad 
    a_3<0.
\]
Throughout this subsection, we express the normal vector $\vb*{n}$ in the same principal-axis basis.

\begin{theorem}
The sign of
\[
    \eta \coloneqq \vb*{n}^{\mathsf T} \vb*{A}^{-1} \vb*{n}
    \label{eq:eta_SM}
\]
determines the intersection pattern:
\[
\begin{array}{rcl}
    \eta<0 & \longrightarrow &
    \text{the plane intersects the cone only at the origin}, \\ [1mm]
    \eta=0 & \longrightarrow &
    \text{the plane is tangent to the cone and gives one tangent line}, \\ [1mm]
    \eta>0 & \longrightarrow &
    \text{the plane cuts the cone and gives two distinct lines}.
\end{array}
\]
\end{theorem}

\begin{proof} 
The relevant object is the dual cone of $P_2(\vb*k)=0$, 
\[ 
    P_{2,d}(\vb*{k})=\vb*{k}^{\mathsf T}\vb*{A}^{-1}\vb*{k}=0. 
\] 
The tangent plane to $P_2(\vb*{k})=0$ at a nonzero point $\vb*{k}_0$ has a normal vector 
\[ 
    \vb*{n}_0=2\vb*{A}\vb*{k}_0. 
\] 
Since $\vb*{k}_0^{\mathsf T}\vb*{A}\vb*{k}_0=0$, this normal vector satisfies $\vb*{n}_0^{\mathsf T}\vb*{A}^{-1}\vb*{n}_0=0$, meaning that ${\vb*n}_0$ is on the dual cone $P_{2,d}=0$. 
Conversely, if a nonzero point $\vb*{n}_0$ satisfies $\vb*{n}_0^{\mathsf T}\vb*{A}^{-1}\vb*{n}_0=0$, then $\vb*{k}_0=\frac{1}{2}\vb*{A}^{-1}\vb*{n}_0$ lies on the original cone, and the plane $\vb*{n}_0^{\mathsf T}\vb*{k}=0$ is tangent to the cone. Thus $\eta=0$ is precisely the tangency condition. 

The dual cone separates the possible normal directions into two regions, $\eta<0$ and $\eta>0$. Within each region, the number of real intersection lines cannot change unless the plane becomes tangent to the cone, which occurs only at $\eta=0$. It remains only to identify which sign corresponds to which geometry, and two representative choices suffice. For $\vb*{n}=(0,0,1)$, we have $\eta=1/a_3<0$, and the plane $k_z=0$ intersects $a_1 k_x^2+a_2 k_y^2=0$ only at the origin. For $\vb*{n}=(1,0,0)$, we have $\eta=1/a_1>0$, and the plane $k_x=0$ gives $a_2 k_y^2+a_3 k_z^2=0$, which consists of two distinct real lines since $a_2>0$ and $a_3<0$. Therefore $\eta<0$ gives only the origin, $\eta=0$ gives one tangent line, and $\eta>0$ gives two distinct lines. 
\end{proof}

\renewcommand{\arraystretch}{1.5}
\begin{table}[h]
\centering
\caption{Summary of the algebraic classification for the linear-quadratic case. The sign of $\eta$, defined in Eq.~\eqref{eq:eta_SM}, determines the intersection geometry. Straight pinch lines appear if and only if $\eta\ge0$.}
\label{tab:LQ_classification}
\begin{ruledtabular}
\begin{tabular}{cc}
\textbf{Algebraic condition} & \textbf{Intersection geometry} \\
\hline 
(a) $\eta<0$ & Origin only \\
(b) $\eta=0$ & One tangent line \\
(c) $\eta>0$ & Two distinct lines \\
\end{tabular}
\end{ruledtabular}
\end{table}
\renewcommand{\arraystretch}{1.0}

\subsection{The quadratic-quadratic case}

We next prove the classification statement for the quadratic-quadratic case, summarized in Fig.~\ref{fig:QQ-flowchart}. 
The two quadratic equations are 
\[
    P_1(\vb*{k})=\vb*{k}^{\mathsf T}\vb*{A}\vb*{k}=0,
    \qquad
    P_2(\vb*{k})=\vb*{k}^{\mathsf T}\vb*{B}\vb*{k}=0,
    \label{eq:SM_QQ_def}
\]
where $\vb*{A}$ and $\vb*{B}$ are real, symmetric, nondegenerate, and indefinite $3\times3$ matrices.
A direct elimination of variables generally leads to a quartic equation.
Instead, we use the following one-parameter family of quadratic forms:
\[
    P_\lambda(\vb*{k}) \coloneqq 
    P_1(\vb*{k})-\lambda P_2(\vb*{k}) =
    \vb*{k}^{\mathsf T}\vb*{M}(\lambda)\vb*{k}, 
    \qquad
    \vb*{M}(\lambda) \coloneqq \vb*{A}-\lambda\vb*{B},
    \qquad \lambda\in\R .
    \label{eq:SM_QQ_pencil}
\]
We take $\lambda$ to be real so that $\vb*{M}(\lambda)$ remains real and symmetric.
For any real $\lambda$, the pair of equations $P_\lambda=0$ and $P_2=0$ has exactly the same common zero set as the original pair in Eq.~\eqref{eq:SM_QQ_def} because $P_1=P_\lambda+\lambda P_2$.
Thus, we may replace $P_1$ by any member $P_\lambda$ of this family without changing the intersection set.

We use this freedom to choose a member whose quadratic form is degenerate.
Such values of $\lambda$ are determined by
\[ 
    \det \vb*{M}(\lambda)=0. 
    \label{eq:SM_QQ_det} 
\] 
Since $\det \vb*{M}(\lambda)$ is a real cubic polynomial, it has at least one real root.
Choose one real root and denote it by $\lambda_0$.
Then, the original intersection problem is equivalently written as
\[ 
    P_{\lambda_0}(\vb*{k})=0, 
    \qquad 
    P_2(\vb*{k})=0. 
    \label{eq:SM_QQ_reduction} 
\]
When Eq.~\eqref{eq:SM_QQ_det} has several real roots, it is sufficient to choose one convenient root.
Different roots give different degenerate quadratic forms, but after imposing $P_2=0$, they all describe the same common zero set.
When a simple real root exists, we choose such a root for convenience.
The remaining case is the triple-root case treated separately.

\begin{lemma}
If $\lambda_0$ is a simple real root of $\det\vb*{M}(\lambda)=0$, then $\rank \vb*{M}(\lambda_0)=2$.
\end{lemma}

\begin{proof}
Since $\det \vb*{M}(\lambda_0)=0$, the matrix $\vb*{M}(\lambda_0)$ is singular; hence $\rank \vb*{M}(\lambda_0) \leq 2$.
For contradiction, suppose that $\rank\vb*{M}(\lambda_0)\leq1$.
For a $3\times3$ matrix, this implies
\[
    \operatorname{adj}\vb*{M}(\lambda_0)=0,
\]
since all cofactors vanish.
Using Jacobi's formula,
\[
    \frac{d}{d\lambda} \det \vb*{M}(\lambda) =
    \tr\!\left(
        \operatorname{adj}\vb*{M}(\lambda)
        \frac{d\vb*{M}(\lambda)}{d\lambda}
    \right) =
    -\tr\!\left(
        \operatorname{adj}\vb*{M}(\lambda)\vb*{B}
    \right),
\]
we obtain
\[
    \left. \frac{d}{d\lambda} \det \vb*{M}(\lambda) \right|_{\lambda=\lambda_0} = 0.
\]
This contradicts the assumption that $\lambda_0$ is a simple root.
Therefore $\rank\vb*{M}(\lambda_0)=2$.
\end{proof}

\noindent
Therefore, the rank must be checked only in the triple-root cases.
In the following, we exclude the rank-zero case because $\vb*{M}(\lambda_0)=0$ would imply $\vb*{A}=\lambda_0\vb*{B}$, so the two constraints in Eq.~\eqref{eq:SM_QQ_def} would not be independent.

We now analyze the simple-root case.
Let $e_1$ and $e_2$ be the two nonzero eigenvalues of $\vb*{M}(\lambda_0)$, and let $\vb*{u}_1$ and $\vb*{u}_2$ be the corresponding orthonormal eigenvectors.
Let $\vb*{v}$ span the kernel of $\vb*{M}(\lambda_0)$.
Then
\[
    P_{\lambda_0}(\vb*{k}) =
    e_1(\vb*{u}_1^{\mathsf T}\vb*{k})^2
    + e_2(\vb*{u}_2^{\mathsf T}\vb*{k})^2.
\]

First, we consider the case $e_1e_2<0$.
Without loss of generality, take $e_1>0$ and $e_2<0$.
Then $P_{\lambda_0}=0$ factorizes as
\[
    P_{\lambda_0}(\vb*{k}) =
    \left(\sqrt{e_1}\,\vb*{u}_1^{\mathsf T}\vb*{k}
    - \sqrt{-e_2}\,\vb*{u}_2^{\mathsf T}\vb*{k}\right)
    \left(\sqrt{e_1}\,\vb*{u}_1^{\mathsf T}\vb*{k}
    + \sqrt{-e_2}\,\vb*{u}_2^{\mathsf T}\vb*{k}\right).
\]
Thus, the zero set of $P_{\lambda_0}=0$ is the union of two real planes,
\[
    \vb*{n}_+^{\mathsf T}\vb*{k}=0,
    \qquad
    \vb*{n}_-^{\mathsf T}\vb*{k}=0,
\]
where $\vb*{n}_\pm=\sqrt{e_1}\,\vb*{u}_1\pm\sqrt{-e_2}\,\vb*{u}_2$ are the normal vectors.
The common zeros of $P_1=0$ and $P_2=0$ are therefore obtained by intersecting each of these planes with the cone $P_2=0$.

Since the equation $P_2=0$ is unchanged by multiplying $P_2$ by an overall minus sign, we choose this sign so that $\vb*{B}$ has two positive eigenvalues and one negative eigenvalue.
Then the criterion in Eq.~\eqref{eq:eta_SM} applies directly to each plane.
We define
\[
    \eta_\pm =
    \vb*{n}_\pm^{\mathsf T} \vb*{B}^{-1} \vb*{n}_\pm .
\]
For each plane, the sign of $\eta_\pm$ determines the intersection with $P_2=0$:
\[
\begin{array}{rcl}
    \eta_\pm < 0
    & \longrightarrow &
    \text{no nonzero real line from that plane}, \\ [1mm]
    \eta_\pm = 0
    & \longrightarrow &
    \text{one tangent line from that plane}, \\ [1mm]
    \eta_\pm > 0
    & \longrightarrow &
    \text{two distinct real lines from that plane}.
\end{array}
\]

\begin{lemma}
Let $\lambda_0$ be a simple real root of $\det\vb*{M}(\lambda)=0$, and let $\vb*{v}$ span $\ker\vb*{M}(\lambda_0)$.
Then
\[
    P_2(\vb*{v}) = \vb*{v}^{\mathsf T}\vb*{B}\vb*{v} \neq 0.
\]
Consequently, if $P_{\lambda_0}=0$ splits into two real planes, the nonzero real lines obtained by intersecting these planes with $P_2=0$ are automatically distinct.
\label{lem:2}
\end{lemma}

\begin{proof}
For a rank-two $3\times3$ matrix, the adjugate matrix has rank one and is supported on the kernel direction.
Hence
\[
    \operatorname{adj}\vb*{M}(\lambda_0)
    =
    c\,\vb*{v}\vb*{v}^{\mathsf T},
    \quad c\neq0.
\]
Using Jacobi's formula at $\lambda=\lambda_0$, we obtain
\[
    \left. \frac{d}{d\lambda} \det \vb*{M}(\lambda) \right|_{\lambda=\lambda_0}
    =
    -\tr\!\left[
        \operatorname{adj}\vb*{M}(\lambda_0)\vb*{B}
    \right]
    =
    -c\,\vb*{v}^{\mathsf T}\vb*{B}\vb*{v}.
\]
If $P_2(\vb*{v})=0$, then the derivative of $\det\vb*{M}(\lambda)$ at $\lambda=\lambda_0$ would vanish.
This contradicts the assumption that $\lambda_0$ is a simple root.
Therefore $P_2(\vb*{v})\neq0$.

Now suppose that a line obtained from the plane $\vb*{n}_+^{\mathsf T}\vb*{k}=0$ coincides with a line obtained from the plane $\vb*{n}_-^{\mathsf T}\vb*{k}=0$.
The coincident line must then lie in the intersection of the two planes.
This intersection is precisely $\ker\vb*{M}(\lambda_0)$.
Hence, the coincident line would have to satisfy $P_2(\vb*{v})=0$, which is impossible for a simple root.
Therefore, the lines obtained from the two planes are distinct.
\end{proof}

\noindent
Thus, in the case $e_1e_2<0$, the full intersection is the union of the lines obtained from the two plane-cone tests above, with no double counting between the two planes.

We next consider the case $e_1e_2>0$.
In this case, the two nonzero eigenvalues have the same sign, so $P_{\lambda_0}(\vb*{k})=0$
forces
\[
    \vb*{u}_1^{\mathsf T}\vb*{k}=0,
    \qquad
    \vb*{u}_2^{\mathsf T}\vb*{k}=0.
\]
Thus, the real zero set of $P_{\lambda_0}=0$ is only the kernel direction spanned by $\vb*{v}$.
However, Lemma~\ref{lem:2} shows that this direction does not satisfy $P_2=0$ for a simple root.
Therefore, a simple root with $e_1e_2>0$ gives only the origin and no nonzero intersection line.

We next move on to the triple-root cases, which require separate treatment.
Let $\lambda_\ast$ be a triple root of $\det\vb*M(\lambda)$. In this case, $\rank\vb*M(\lambda_\ast)$ is either two or one, so we analyze the two possibilities separately.

\begin{lemma}
Suppose that $\rank \vb*{M}(\lambda_\ast)=2$.
Then $\vb*{M}(\lambda_\ast)$ is indefinite, and $P_{\lambda_\ast}=0$ splits into two real planes. One of them is tangent to the cone $P_2=0$, while the other is secant. The two plane-cone intersections share the kernel line of $\vb*{M}(\lambda_\ast)$. After identifying this common line, the common zero set of $P_1=0$ and $P_2=0$ consists of two distinct real lines, including one tangent line.
\end{lemma}

\begin{proof}
We introduce $\lambda=\lambda_\ast+\mu$ ($\mu\in\R$), and then $\vb*{M}(\lambda_\ast+\mu) = \vb*{M}(\lambda_\ast) - \mu \vb*{B}$.
We first show that $\vb*{M}(\lambda_\ast)$ is indefinite. 
For contradiction, suppose that it is semidefinite. 
Since $\vb*{M}(\lambda_\ast)$ is real symmetric and has rank two, we may choose coordinates such that
\[
    \vb*{M}(\lambda_\ast) = \operatorname{diag}(1,1,0).
\]
Writing
\[
    \vb*{B}=
    \begin{pmatrix}
        \alpha & \delta & \rho \\
        \delta & \beta  & \sigma \\
        \rho   & \sigma & \gamma
    \end{pmatrix},
\]
we expand
\[
    \det(\vb*{M}(\lambda_\ast) - \mu \vb*{B})
    = - \gamma \mu +
    (\alpha\gamma+\beta\gamma-\rho^2-\sigma^2)\mu^2
    +O(\mu^3).
\]
Since $\lambda_\ast$ is a triple root, the coefficients of $\mu$ and $\mu^2$ must vanish. 
Thus
\[
    \gamma=0,
    \qquad
    \rho=\sigma=0.
\]
But then the third row and the third column of $\vb*{B}$ vanish, so $\vb*{B}$ is singular. 
This contradicts the assumption that $\vb*{B}$ is nondegenerate. 
Therefore $\vb*{M}(\lambda_\ast)$ must be indefinite.

We now choose coordinates $(k'_x,k'_y,k'_z)$ so that the two real planes are $k'_x=0$ and $k'_y=0$. 
In these coordinates,
\[
    \vb*{M}(\lambda_\ast) =
    \begin{pmatrix}
        0 & 1/2 & 0 \\
        1/2 & 0 & 0 \\
        0 & 0 & 0
    \end{pmatrix}.
\]
The common line of the two planes is $k'_x=k'_y=0$.
In the same coordinates, we write
\[
    \vb*{B} =
    \begin{pmatrix}
        a & d & f \\
        d & b & g \\
        f & g & h
    \end{pmatrix}.
\]
Expanding the characteristic polynomial around the triple root gives
\[
    \det\!\left[\vb*{M}(\lambda_\ast)-\mu\vb*{B}\right] =
    \frac{h}{4}\mu+(fg-dh)\mu^2+O(\mu^3).
\]
Since $\lambda_\ast$ is a triple root, the coefficients of $\mu$ and $\mu^2$ must vanish. Therefore
\[
    h=0, \qquad fg=0.
\]
The first condition $h=0$ has a direct geometric meaning. 
On the common line, we have $k'_x=k'_y=0$, and hence
\[
    P_2 = h (k'_z)^2.
\]
Thus $h=0$ means that the kernel line lies on the cone $P_2=0$.
The second condition $fg=0$, together with the nondegeneracy of $\vb*{B}$, implies that exactly one of $f$ and $g$ is zero. 

We now restrict $P_2$ to the two planes of $P_{\lambda_\ast}=0$. 
On the plane $k'_x=0$,
\[
    P_2 = b (k'_y)^2+2g k'_y k'_z = k'_y(bk'_y+2gk'_z),
\]
where we used $h=0$. 
Here, the factor $k'_y=0$ is precisely the shared line. 
If $g=0$, this restriction becomes $b (k'_y)^2$, so the intersection is the double line; hence, the plane $k'_x=0$ is tangent to the cone $P_2=0$. 
If $g\neq0$, the restriction gives two distinct real lines, $k'_x=k'_y=0$ and $bk'_y+2gk'_z=0$; hence, the plane $k'_x=0$ is secant.
Similarly, on the plane $k'_y=0$,
\[
    P_2 = a (k'_x)^2+2f k'_x k'_z = k'_x(ak'_x+2fk'_z).
\]
Therefore, this plane is tangent if $f=0$, and secant if $f\neq0$.
Since exactly one of $f$ and $g$ is zero, exactly one of the two planes is tangent to $P_2=0$ along the kernel line, while the other is secant and gives the kernel line together with one additional real line. 
After identifying the shared line, the total common zero set consists of two distinct real lines, one of which is the tangent line.
\end{proof}

\begin{lemma}
Suppose that $\rank \vb*{M}(\lambda_\ast)=1$.
The common zero set of $P_1=0$ and $P_2=0$ consists of one tangent line.
\end{lemma}

\begin{proof} 
Since $\vb*{M}(\lambda_\ast)$ is real symmetric and has rank one, it can be written as 
\[ 
    \vb*{M}(\lambda_\ast) = e \, \vb*{n} \vb*{n}^{\mathsf T}, 
    \quad 
    e \neq 0, 
\] 
for some nonzero vector $\vb*{n}$. Therefore, 
\[  
    P_{\lambda_\ast}(\vb*{k}) = \vb*{k}^{\mathsf T} \vb*{M}(\lambda_\ast) \vb*{k} = e(\vb*{n}^{\mathsf T} \vb*{k})^2, 
\] 
and the equation $P_{\lambda_\ast}=0$ gives the double plane $\vb*{n}^{\mathsf T}\vb*{k}=0$. 
We now write $\lambda=\lambda_\ast+\mu$ ($\mu\in\R$). Then 
\[ 
    \vb*{M}(\lambda_\ast+\mu) = e\,\vb*{n}\vb*{n}^{\mathsf T} - \mu\vb*{B}. 
\] 
Using the matrix determinant lemma, 
\[ 
    \begin{aligned} 
        \det\vb*{M}(\lambda_\ast+\mu) 
        & = \det(-\mu\vb*{B}) \left( 1 + e \, \vb*{n}^{\mathsf T}(-\mu\vb*{B})^{-1}\vb*{n} \right) \\ 
        & = - \mu^3 \det\vb*{B} + e \, \mu^2 (\det\vb*{B}) \, \vb*{n}^{\mathsf T} \vb*{B}^{-1} \vb*{n}. 
    \end{aligned} 
\] 
Since $\lambda_\ast$ is a triple root, the expansion of $\det\vb*{M}(\lambda_\ast+\mu)$ must start at order $\mu^3$.
Therefore, the coefficient of $\mu^2$ must vanish. Due to $e\neq0$ and $\det\vb*{B}\neq0$, this gives 
\[ 
    \vb*{n}^{\mathsf T}\vb*{B}^{-1}\vb*{n}=0.
\] 
By the criterion in Eq.~\eqref{eq:eta_SM}, this means that the plane $\vb*{n}^{\mathsf T}\vb*{k}=0$ is tangent to the cone $P_2=0$, and their intersection consists of one real tangent line. 
\end{proof}
\noindent
Combining the two lemmas, the exceptional triple-root case is completely determined by $\rank\vb*{M}(\lambda_\ast)$. If $\rank\vb*{M}(\lambda_\ast)=1$, the intersection consists of one tangent line. If $\rank\vb*{M}(\lambda_\ast)=2$, the intersection consists of two distinct real lines, including one tangent line. 

\begin{figure}[h]
\centering
\resizebox{\linewidth}{!}{\begin{tikzpicture}[
    >=Latex,
    line width=0.9pt,
    font=\small,
    proc/.style={
        draw,
        rounded corners=2pt,
        align=center,
        inner xsep=5pt,
        inner ysep=5pt,
        minimum height=8mm
    },
    decision/.style={
        draw,
        diamond,
        aspect=2.2,
        align=center,
        inner sep=1.5pt
    },
    arrow/.style={->, line width=0.9pt},
    lab/.style={fill=white, inner sep=1pt, font=\scriptsize}
]

\definecolor{flowblue}{RGB}{220,238,255}
\definecolor{flowyellow}{RGB}{255,241,184}
\definecolor{flowgreen}{RGB}{223,243,223}
\definecolor{flowgreenlight}{RGB}{237,248,237}
\definecolor{floworange}{RGB}{255,226,204}
\definecolor{floworangelight}{RGB}{255,240,227}
\definecolor{flowviolet}{RGB}{235,230,255}

\node[proc, fill=flowblue, text width=40mm] (start) at (0,0)
{Two constraints:\\[1mm]
$P_1(\vb*{k})=0,\quad P_2(\vb*{k})=0$};

\node[proc, fill=flowblue, text width=58mm] (pencil) at (0,-1.65)
{One-parameter family:\\[1mm]
$\vb*{M}(\lambda)=\vb*{A}-\lambda\vb*{B}$};

\draw[arrow] (start.south) -- (pencil.north);

\node[decision, fill=flowyellow, text width=38mm] (simple) at (0,-4.0)
{Does $\det \vb*{M}(\lambda)=0$ have\\
a simple real root $\lambda_0$?};

\draw[arrow] (pencil.south) -- (simple.north);

\node[proc, fill=flowgreenlight, text width=44mm] (choose) at (-5.3,-5.2)
{Choose $\lambda_0$\\
$\rank\mathbf{M}(\lambda_0)=2$};

\node[decision, fill=flowgreen, text width=28mm] (sign) at (-5.3,-7.05)
{Is $e_1e_2<0$?};

\node[proc, fill=flowgreenlight, text width=28mm] (origin) at (-8.2,-8.35)
{origin only};

\node[proc, fill=flowgreenlight, text width=42mm] (planes) at (-2.5,-8.45)
{$P_{\lambda_0}=0$ splits into\\
two real planes};

\node[proc, fill=flowgreenlight, text width=42mm] (eta) at (-2.5,-9.9)
{For each plane $\vb*{n}_{\pm}^{\mathsf T}\vb*{k}=0$,\\
evaluate $\eta_{\pm}=\vb*{n}_{\pm}^{\mathsf T}\vb*{B}^{-1}\vb*{n}_{\pm}$};

\node[proc, fill=flowgreenlight, text width=42mm] (etaout) at (-2.5,-11.82)
{$\eta_{\pm}<0$: no real line,\\
$\eta_{\pm}=0$: one tangent line,\\
$\eta_{\pm}>0$: two distinct lines.\\[1mm]
Total: sum over $\pm$ planes};

\node[proc, fill=floworangelight, text width=38mm] (triple) at (5.3,-5.20)
{triple root $\lambda_\ast$};

\node[decision, fill=floworange, text width=32mm] (rank) at (5.3,-7.10)
{Rank of $\vb*{M}(\lambda_\ast)$?};

\node[proc, fill=floworangelight, text width=38mm] (ranktwo) at (2.5,-9.05)
{two distinct real lines,\\
including one tangent line};

\node[proc, fill=floworangelight, text width=35mm] (rankone) at (8.1,-9.05)
{one tangent line};

\draw[arrow] (triple.south) -- (rank.north);

\draw[arrow]
    (rank.west)
    -- node[lab, above] {rank $2$}
    (ranktwo.north |- rank.west)
    -- (ranktwo.north);

\draw[arrow]
    (rank.east)
    -- node[lab, above] {rank $1$}
    (rankone.north |- rank.east)
    -- (rankone.north);

\draw[arrow]
    (simple.west)
    -- node[lab, above] {Yes}
    (choose.north |- simple.west)
    -- (choose.north);

\draw[arrow]
    (simple.east)
    -- node[lab, above] {No}
    (triple.north |- simple.east)
    -- (triple.north);

\draw[arrow] (choose) -- (sign);

\draw[arrow]
    (sign.west)
    -- node[lab, above] {No}
    (origin.north |- sign.west)
    -- (origin.north);

\draw[arrow]
    (sign.east)
    -- node[lab, above] {Yes}
    (planes.north |- sign.east)
    -- (planes.north);

\draw[arrow] (planes) -- (eta);
\draw[arrow] (eta) -- (etaout);

\end{tikzpicture}}
\caption{
Flowchart for the quadratic-quadratic classification. 
Starting from the one-parameter family $\vb*{M}(\lambda)=\vb*{A}-\lambda\vb*{B}$, the classification is determined by whether $\det\vb*{M}(\lambda)=0$ has a simple real root. 
If a simple real root $\lambda_0$ exists, the intersection geometry is obtained from the signature of $\vb*{M}(\lambda_0)$ and the plane-cone tests for $\eta_{\pm}=\vb*{n}_{\pm}^{\mathsf T}\vb*{B}^{-1}\vb*{n}_{\pm}$. 
If no simple real root exists, the cubic has a triple root $\lambda_\ast$, and the remaining classification is determined by $\rank\vb*{M}(\lambda_\ast)$.}
\label{fig:QQ-flowchart}
\end{figure}
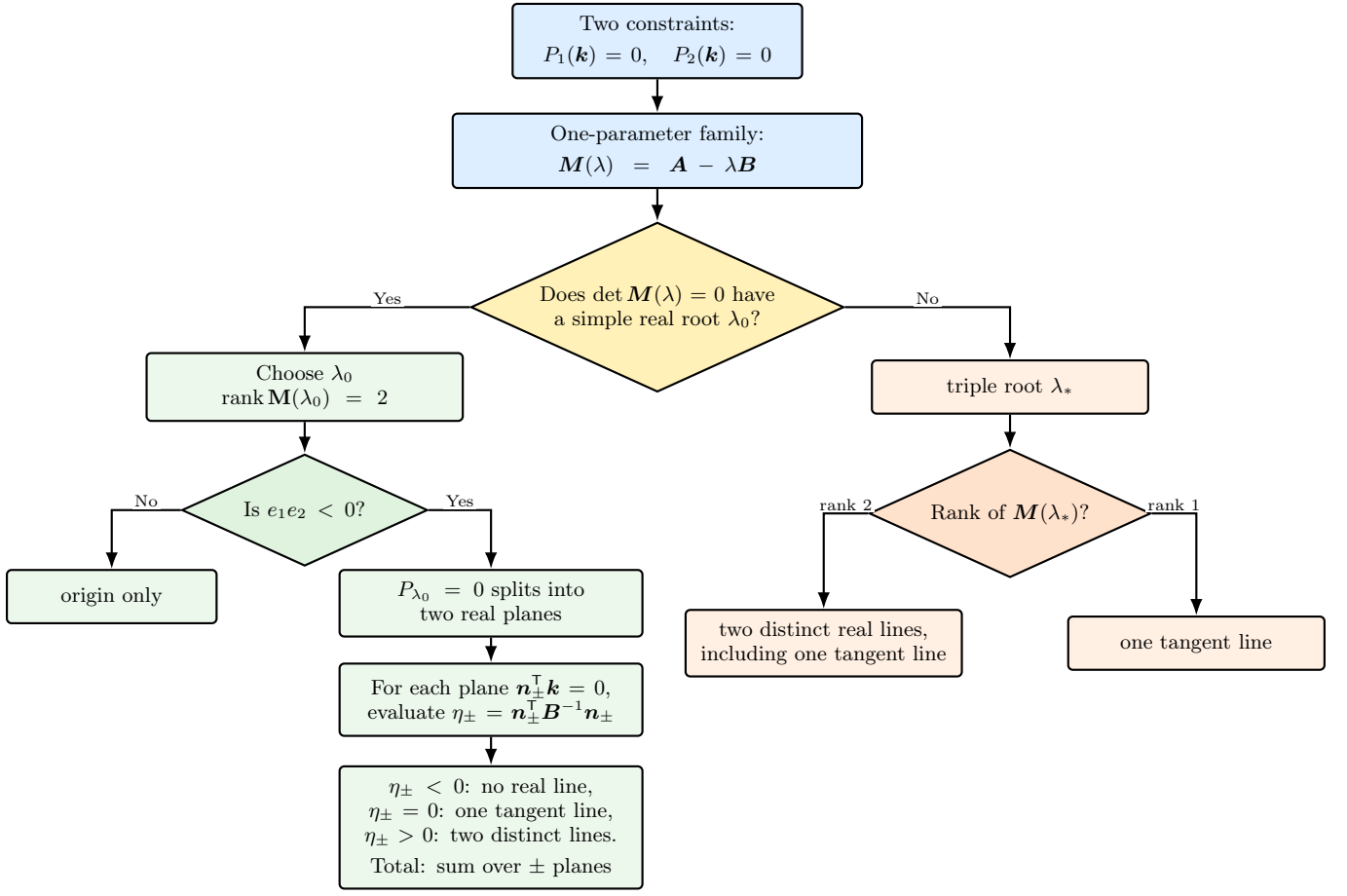

\subsection{The linear-(linear-plus-cubic) case}

We prove the classification statement for the linear-(linear-plus-cubic) case.
The two equations are
\[
   P_1(\vb*{k}) = 0, 
   \qquad
   P_2(\vb*{k}) = c(\vb*{k}) + \ell(\vb*{k}) = 0.
\]
Here $P_1(\vb*{k})$ and $\ell(\vb*{k})$ are homogeneous linear polynomials, and $c(\vb*{k})$ is a homogeneous cubic polynomial.
Since $P_1(\vb*{k})=0$ defines a plane through the origin, we can choose coordinates such that $P_1(\vb*{k}) = k_z$.
The intersection problem is then reduced to a plane curve on $k_z=0$. 
We define the restricted polynomials
\[
    c_3(k_x,k_y) \coloneqq c(k_x,k_y,0),
    \qquad
    \ell_1(k_x,k_y) \coloneqq \ell(k_x,k_y,0),
\]
and write $P_2(k_x,k_y)=c_3(k_x,k_y)+\ell_1(k_x,k_y)$.

Now we assume that \(\ell_1\neq0\), so that the restriction is genuinely of linear-plus-cubic form. 
We can make an invertible linear change of coordinates within the plane, and after relabeling the new coordinates as $k_x$ and $k_y$, set $\ell_1(k_x,k_y)=k_x$.
The cubic part is then a general homogeneous binary cubic, 
\[ 
    c_3(k_x,k_y) = a k_x^3+b k_x^2k_y+c k_xk_y^2+d k_y^3, 
\] 
and the plane curve takes the normal form 
\[
    P_2(k_x,k_y) = k_x \left( 1+a k_x^2+b k_xk_y+c k_y^2 \right) + d k_y^3. 
    \label{seqn.normal.form}
\]

We first show that $d=0$ is the necessary and sufficient condition for the curve to be reducible; equivalently, $d\neq0$ is the condition for irreducibility.

If $d=0$, then
\begin{equation}
    P_2(k_x,k_y)=k_x\left(1+a k_x^2+bk_xk_y+ck_y^2\right),
\end{equation}
so the curve is reducible. Conversely, suppose that $P_2$ is reducible over $\mathbb R$. Since $P_2$ has degree three, it has a real linear factor.

First consider the case in which the linear factor passes through the origin. Since the degree-one part of $P_2$ is exactly $k_x$, any such linear factor must divide $k_x$. Hence the factor is proportional to $k_x$. This forces the cubic homogeneous part
\begin{equation}
    a k_x^3+b k_x^2k_y+ck_xk_y^2+d k_y^3
\end{equation}
to be divisible by $k_x$ and therefore $d=0$.

It remains to exclude the possibility of a linear factor not passing through the origin. Write such a factorization as
\begin{equation}
    P_2(k_x,k_y)=(pk_x+qk_y+r)
    \left(Ak_x^2+Bk_xk_y+Ck_y^2+Dk_x+Ek_y+F\right),
    \qquad r\neq 0 .
\end{equation}
Since $P_2$ has no constant term, the constant term of the product gives
\begin{equation}
    rF=0,
\end{equation}
and hence $F=0$. Matching the degree-one terms gives
\begin{equation}
    rD=1,\qquad rE=0,
\end{equation}
so
\begin{equation}
    D=\frac{1}{r},\qquad E=0 .
\end{equation}
Now the quadratic terms of $P_2$ vanish identically. In particular, the coefficient of $k_y^2$ in the product is
\begin{equation}
    rC+qE=rC,
\end{equation}
and therefore $C=0$. But then, the coefficient of $k_y^3$ in the product is
\begin{equation}
    qC=0.
\end{equation}
Thus the coefficient of $k_y^3$ in $P_2$ must vanish, namely $d=0$.

We have shown that every real factorization of $P_2$ forces $d=0$. Therefore
\begin{equation}
    d\neq 0
    \quad\Longrightarrow\quad
    P_2(k_x,k_y)\text{ is irreducible over }\mathbb R .
\end{equation}
Together with the obvious factorization for $d=0$, this proves
\begin{equation}
    P_2(k_x,k_y)\text{ is reducible over }\mathbb R
    \quad\Longleftrightarrow\quad
    d=0.
\end{equation}

Next, before classifying the reducible and irreducible cases separately, we record a useful structural property: finite singularities can occur only in the reducible case.
\begin{theorem}
For the normal form in Eq.~\eqref{seqn.normal.form}, if the curve $P_2(k_x,k_y)=0$ has a finite singular point, such as a self-intersection or cusp, then the polynomial $P_2(k_x,k_y)$ is reducible.
\label{thm:2}
\end{theorem}

\begin{proof}
Let $(k_{x_0},k_{y_0})$ be a singular point of the curve $P_2(k_x,k_y)=0$. 
Then
\[
    P_2(k_{x_0},k_{y_0})=0,
    \qquad
    \partial_x P_2(k_{x_0},k_{y_0})=0,
    \qquad
    \partial_y P_2(k_{x_0},k_{y_0})=0 .
\]
From the last two equations, we have
\[
    \partial_x c_3(k_{x_0},k_{y_0})=-\partial_x \ell_1(k_{x_0},k_{y_0}),
    \qquad
    \partial_y c_3(k_{x_0},k_{y_0})=-\partial_y \ell_1(k_{x_0},k_{y_0}).
\]
Using Euler's theorem for homogeneous polynomials,
\[
    k_{x_0}\partial_x c_3(k_{x_0},k_{y_0})+k_{y_0}\partial_y c_3(k_{x_0},k_{y_0})=3c_3(k_{x_0},k_{y_0}),
    \quad
    k_{x_0}\partial_x \ell_1(k_{x_0},k_{y_0})+k_{y_0}\partial_y \ell_1(k_{x_0},k_{y_0})=\ell_1(k_{x_0},k_{y_0}),
\]
we obtain
\[
    3c_3(k_{x_0},k_{y_0}) = -\ell_1(k_{x_0},k_{y_0}).
\]
On the other hand, $P_2(k_{x_0},k_{y_0})=0$ gives $c_3(k_{x_0},k_{y_0})+\ell_1(k_{x_0},k_{y_0})=0$.
Combining these two relations yields
\[
    c_3(k_{x_0},k_{y_0})=0,
    \qquad
    \ell_1(k_{x_0},k_{y_0})=0.
\]

Since $\ell_1(k_x,k_y)=k_x$, the singular point satisfies $k_{x_0}=0$. 
Moreover, it cannot be the origin, because $\nabla P_2(0,0)=\nabla \ell_1\neq0$, whereas a singular point requires $\nabla P_2=0$. 
Thus $k_{y_0}\neq0$.
Since $c_3(k_{x_0},k_{y_0})=0$ and $k_{x_0}=0$, we have $d k_{y_0}^3=0$.
Since $k_{y_0}\neq0$, this implies $d=0$.
Therefore $c_3(k_x,k_y)=k_x(ak_x^2+bk_xk_y+ck_y^2)$, and hence
\[
    P_2(k_x,k_y) = k_x\left(ak_x^2+bk_xk_y+ck_y^2+1\right).
    \label{eq:reducible}
\]
Thus $\ell_1(k_x,k_y)=k_x$ is a factor of $P_2(k_x,k_y)$, so $P_2(k_x,k_y)$ is reducible.
\end{proof}

For the reducible cases, the remaining task is to classify the quadratic part in Eq.~\eqref{eq:reducible}.

\begin{lemma}
Let
\[
    Q(k_x,k_y)=ak_x^2+bk_xk_y+ck_y^2,
    \quad
    \vb*{Q} =
    \begin{pmatrix}
        a & b/2 \\
        b/2 & c
    \end{pmatrix}.
\]
Then, the factor $Q(k_x,k_y)+1=0$ is classified as follows:
\[
\begin{array}{ccl}
    \vb*{Q} \ge 0
    \longrightarrow 
    \text{no real point}, \quad
    \vb*{Q} < 0
    \longrightarrow
    \text{an ellipse}, \\ [1mm]
    \vb*{Q} \, \text{indefinite}
    \longrightarrow
    \text{a hyperbola}, \quad
    \vb*{Q} \le 0 ,\ \rank\vb*{Q}=1
    \longrightarrow 
    \text{two parallel real lines}.
\end{array}
\]
Consequently, the reducible curve $P_2(k_x,k_y)=0$ is the union of the line $k_x=0$ and one of the following real conics: the empty set, an ellipse, a hyperbola, or a pair of parallel real lines. In the hyperbola case, the line $k_x=0$ intersects the hyperbola at two finite real points if and only if $c<0$. If $c>0$, it has no finite intersection and is not parallel to either asymptotic direction. If $c=0$, it has no finite intersection but is parallel to one asymptotic direction of the hyperbola. In the rank-one negative-semidefinite case, the line $k_x=0$ intersects the parallel pair if and only if $c<0$. If $c=0$, all three real lines are mutually parallel; this degenerate subcase has $Q=a k_x^2$ with $a<0$, so $P_2$ depends only on $k_x$ and is outside the nondegenerate two-variable setting considered for the representative physical examples.
\end{lemma}

\begin{proof}
The curve $P_2(k_x,k_y)=0$ is the union of the line $k_x=0$ and the conic $Q(k_x,k_y)+1=0$.
Therefore, it is enough to classify the real conic $Q+1=0$.
Since $\vb*{Q}$ is a real symmetric matrix, it can be diagonalized by an orthogonal transformation in the $(k_x,k_y)$-plane. 
We denote the coordinates after this transformation by $(k'_x,k'_y)$.
The shape of the conic is then determined by the signs of the eigenvalues of $\vb*{Q}$.

If \(\vb*{Q}\) is positive semidefinite, then
\[
    Q(k_x,k_y)+1>0
\]
for all $(k_x,k_y)$, so the conic $Q+1=0$ has no real points. Therefore, the full curve is only the line $k_x=0$.

If $\vb*{Q}$ is negative definite, in the transformed coordinates,
\[
    Q=-\alpha (k'_x)^2-\beta (k'_y)^2,
    \qquad
    \alpha,\beta>0.
\]
Therefore
\[
    Q+1=0
    \quad\Longleftrightarrow\quad
    \alpha (k'_x)^2+\beta (k'_y)^2=1,
\]
which is an ellipse.

If \(\vb*{Q}\) is indefinite, in the transformed coordinates,
\[
    Q=\alpha (k'_x)^2-\beta (k'_y)^2,
    \qquad
    \alpha,\beta>0.
\]
Thus
\[
    Q+1=0
    \quad\Longleftrightarrow\quad
    \beta (k'_y)^2-\alpha (k'_x)^2=1,
\]
which is a hyperbola. 
To determine whether this hyperbola intersects the line $k_x=0$, we return to the original coordinates and substitute $k_x=0$ into the conic:
\[
    Q(0,k_y)+1=ck_y^2+1=0.
\]
The line $k_x=0$ intersects the hyperbola in two finite real points if $c<0$. If $c>0$, it has no finite intersection and its direction is not an asymptotic direction of the hyperbola. If $c=0$, it has no finite intersection, and its direction is an asymptotic direction because $Q(0,1)=0$.

Finally, suppose that $\vb*{Q}$ is negative semidefinite of rank one. 
Then, in the transformed coordinates,
\[
    Q=-\alpha (k'_x)^2,
    \qquad
    \alpha>0.
\]
Hence
\[
    Q+1=0
    \quad\Longleftrightarrow\quad
    k'_x=\pm\frac{1}{\sqrt{\alpha}},
\]
which is a pair of parallel real lines. 
The full curve is therefore the union of the line $k_x=0$ and this parallel pair. 
To determine whether $k_x=0$ intersects the pair, we substitute $k_x=0$ into the conic:
\[
    Q(0,k_y)+1=ck_y^2+1=0.
\]
Since $\vb*{Q}\le0$, one has $c\le0$. 
If $c<0$, there are two finite real intersections. 
If $c=0$, there is no finite intersection; in fact the rank-one condition then forces $b=0$, so $Q=a k_x^2$ with $a<0$. 
Thus $P_2=k_x(1+a k_x^2)$, and the three real lines are mutually parallel. 
This last subcase contains no $k_y$-dependence and is excluded from the nondegenerate two-variable physical examples.
\end{proof}

We next consider the irreducible case $d\neq0$. 
Theorem~\ref{thm:2} implies that the curve has no singular point at a finite $(k_x,k_y)$.
Thus, the remaining distinction comes from the large-$|k_x|$ branch structure.
Since $d\neq0$, every nonzero point of the curve has $k_x\neq0$. 
Indeed, if $k_x=0$, then $P_2(0,k_y)=d k_y^3$, so $P_2=0$ implies $k_y=0$. 
For nonzero points, we can therefore introduce the slope variable $t=k_y/k_x$ and write 
\[ 
    c_3(k_x,k_y)=k_x^3 f(t), 
    \qquad 
    f(t)=a+bt+ct^2+dt^3. 
\]  
The asymptotic directions are obtained by studying the limit $|k_x|\to\infty$. Introduce $z=1/k_x$. Then $|k_x|\to\infty$ corresponds to $z\to0$, and the equation $P_2=0$ becomes 
\[ 
    f(t)+z^2=0. \label{eq:linear_cubic_infinity_local} 
\]
Therefore, if a real branch approaches a finite limiting slope $t\to r$ at large $|k_x|$, then $r$ must satisfy 
\[ 
    f(r)=0. 
\] 
Thus, the possible real asymptotic directions are determined by the real roots of $f(t)$, and the root multiplicity determines the local branch type.

\begin{lemma}
Let $r$ be a simple real root of $f(t)$, namely $f(r)=0$.
Then the curve $P_2(k_x,k_y)=0$ has a real branch with an asymptotic direction $k_y/k_x\to r$ as $|k_x|\to\infty$. 
Moreover, this branch has the ordinary affine asymptote
\[
    k_y=rk_x.
\]
\label{lem:simple_root_asymptote}
\end{lemma}
\begin{proof}
Since $r$ is a simple root, $f'(r)\neq0$. 
By the implicit function theorem, there is a real analytic solution $t=t(z)$ near $z=0$ with $t(0)=r$. 
Expanding,
\[
    t(z) = r-\frac{z^2}{f'(r)}+O(z^4).
\]
Returning to affine coordinates, $z=1/k_x$ and $t=k_y/k_x$, so
\[
    k_y = rk_x-\frac{1}{f'(r)k_x}+O(k_x^{-3}).
\]
Hence
\[
    k_y-rk_x\to0
\]
as $|k_x|\to\infty$, which means that the branch has the affine asymptote $k_y=rk_x$.
\end{proof}
Therefore, each simple root of $f(t)$ gives one ordinary affine asymptote. 
Multiple roots require additional analysis, because a multiple root may give two asymptotes, no real branch, or a branch with only an asymptotic direction.

\begin{lemma}
Let $r$ be a real double root of $f(t)$, namely $f(r)=0$, $f'(r)=0$, and $f''(r)\neq0$.
Then the real branches of $P_2=0$ approaching the large-$|k_x|$ direction $k_y/k_x\to r$ are classified by the sign of $f''(r)$. 
If $f''(r)<0$, then two real branches approach this direction, and their affine asymptotes are
\[
    k_y-rk_x=\pm\left(-\frac{1}{2}f''(r)\right)^{-1/2}.
\]
If $f''(r)>0$, then no real branch approaches this direction.
\label{lem:double_root_asymptote}
\end{lemma}

\begin{proof}
Set $u=t-r$.
Since $r$ is a double root,
\[
    f(r+u)=\frac{1}{2} f''(r) u^2+O(u^3).
\]
Therefore, the local equation near $(u,z)=(0,0)$ becomes
\[
    \frac{1}{2} f''(r) u^2+z^2+O(u^3)=0.
\]

If $f''(r)<0$, the leading equation
\[
    \frac{1}{2} f''(r) u^2+z^2=0
\]
has two real branches,
\[
    u=\pm\left(-\frac{1}{2}f''(r)\right)^{-1/2}z+O(z^2).
\]
Returning to the original variables, we find that the two real branches have affine asymptotes
\[
    k_y-rk_x=\pm\left(-\frac{1}{2}f''(r)\right)^{-1/2}.
\]

If $f''(r)>0$, then the quadratic part $f''(r) u^2/2+z^2$ is positive definite. 
The higher-order term \(O(u^3)\) is negligible in a sufficiently small neighborhood of \((u,z)=(0,0)\), so the equation has no nontrivial real solution near this point. 
Therefore, no real branch approaches the direction \(k_y/k_x\to r\).
\end{proof}

\begin{lemma}
Let $r$ be a real triple root of $f(t)$. 
Since $f(t)$ is cubic, this means
\[
    f(t)=d(t-r)^3,
    \qquad
    d\neq0 .
\]
Then the curve $P_2(k_x,k_y)=0$ has a real branch with asymptotic direction $k_y/k_x\to r$ as $|k_x|\to\infty$. 
Moreover, along this branch, $k_y-rk_x$ does not approach a finite constant. 
Therefore, the branch has a real asymptotic direction, but no finite-offset affine asymptote of slope $r$.
\label{lem:triple_root_asymptotic_direction}
\end{lemma}

\begin{proof}
Set $u=t-r$, and then $f(t)=d(t-r)^3=d u^3$.
Therefore, the local equation becomes $d u^3+z^2=0$.
For real $z$ sufficiently close to zero, this equation has a real solution $u=-\operatorname{sgn}(d)\,|d|^{-1/3}|z|^{2/3}$.
Hence $u\to0$ as $z\to0$. 
Therefore, the curve has a real asymptotic direction $k_y/k_x\to r$.

However, $k_y-rk_x=u k_x$.
Using $z=1/k_x$ and $u\sim |z|^{2/3}$, we obtain
\[
    |k_y-rk_x|
    \sim
    |z|^{2/3}|z|^{-1}
    =
    |z|^{-1/3}
    =
    |k_x|^{1/3}.
\]
Therefore $k_y-rk_x$ does not approach a finite constant as $|k_x|\to\infty$. 
The branch has a real asymptotic direction, but not in general a finite-offset affine asymptote.
\end{proof}
Combining Lemmas~\ref{lem:simple_root_asymptote}, \ref{lem:double_root_asymptote}, and \ref{lem:triple_root_asymptotic_direction}, the large-$|k_x|$ behavior in the irreducible case is completely determined by the real roots of $f(t)$ and their multiplicities. 
A simple real root gives one ordinary affine asymptote. 
A double real root gives two affine asymptotes when $f''(r)<0$, while it gives no real branch approaching that direction when $f''(r)>0$. 
A triple real root gives one real asymptotic direction, but no finite-offset affine asymptote.
Together with the reducible cases classified above, this completes the classification of the linear-(linear-plus-cubic) case summarized in Fig.~\ref{fig:linear-linear-plus-cubic-flowchart}.

\begin{figure}[h]
\centering
\resizebox{0.98\linewidth}{!}{\begin{tikzpicture}[
    >=Latex,
    line width=0.9pt,
    font=\small,
    proc/.style={
        draw,
        rounded corners=2pt,
        align=center,
        inner xsep=5pt,
        inner ysep=5pt,
        minimum height=8mm
    },
    decision/.style={
        draw,
        diamond,
        aspect=2.2,
        align=center,
        inner sep=1.5pt
    },
    arrow/.style={->, line width=0.9pt},
    lab/.style={fill=white, inner sep=1pt, font=\scriptsize}
]

\definecolor{flowblue}{RGB}{220,238,255}
\definecolor{flowyellow}{RGB}{255,241,184}
\definecolor{flowgreen}{RGB}{223,243,223}
\definecolor{flowgreenlight}{RGB}{237,248,237}
\definecolor{floworange}{RGB}{255,226,204}
\definecolor{floworangelight}{RGB}{255,240,227}
\definecolor{flowviolet}{RGB}{235,230,255}

\node[proc, fill=flowblue, text width=42mm] (start) at (0,0)
{Two constraints:\\[1mm]
$P_1(\vb*{k})=0,\quad P_2(\vb*{k})=0$};

\node[proc, fill=flowblue, text width=42mm] (plane) at (0,-1.7)
{Choose coordinates so that\\[1mm]
$P_1(\vb*{k})=k_z$};

\node[proc, fill=flowblue, text width=70mm] (curve) at (0,-3.9)
{Restrict to $k_z=0$ and choose in-plane coordinates so that\\[1mm]
$\ell_1(k_x,k_y)=k_x$,\\[1mm]
$P_2(k_x,k_y)=k_x\!\left(1+a k_x^2+b k_xk_y+c k_y^2\right)+d k_y^3$};

\node[decision, fill=flowyellow, text width=30mm] (dzero) at (0,-6.4)
{Is $d=0$?};

\draw[arrow] (start.south) -- (plane.north);
\draw[arrow] (plane.south) -- (curve.north);
\draw[arrow] (curve.south) -- (dzero.north);

\node[proc, fill=flowgreenlight, text width=60mm] (redform) at (-6.2,-8.0)
{Reducible case:\\[1mm]
$P_2=k_x\{Q(k_x,k_y)+1\}$,\\[1mm]
the curve is the union of\\[1mm] 
$k_x=0$ and $Q(k_x,k_y)+1=0$.};

\node[proc, fill=flowgreenlight, text width=86mm] (qclass) at (-6.2,-11.4)
{Classify the shape of $Q(k_x,k_y)+1=0$:\\[1mm]
$\mathbf{Q}\ge0$: a single line,\\[1mm]
$\mathbf{Q}<0$: a line + an ellipse\\[1mm]
$\mathbf{Q}$ indefinite: a line + hyperbola\\[1mm]
\hspace{4mm}($c<0$: two intersections; $c>0$: no finite intersection),\\[1mm]
\hspace{4mm}($c=0$: no finite intersection, parallel to one asymptotic direction),\\[1mm]
$\mathbf{Q}\le0$, $\operatorname{rank}\mathbf{Q}=1$: a line + two parallel lines\\[1mm]
\hspace{4mm}($c<0$: two intersections; $c=0$: three parallel lines).};

\draw[arrow]
    (dzero.west)
    -- node[lab, above] {Yes}
    (redform.north |- dzero.west)
    -- (redform.north);

\draw[arrow] (redform.south) -- (qclass.north);

\node[proc, fill=floworangelight, text width=55mm] (slope) at (6.2,-7.7)
{Irreducible case:\\[1mm]
set $t=k_y/k_x$,\\[1mm]
$f(t)=a+bt+ct^2+dt^3$};

\node[decision, fill=floworange, text width=34mm] (mult) at (6.2,-10.3)
{Multiplicity of a real root $r$ of $f(t)$?};

\node[proc, fill=floworangelight, text width=34mm] (simple) at (2.2,-12.4)
{One real branch,\\
asymptote $k_y=rk_x$};

\node[proc, fill=floworangelight, text width=46mm] (triple) at (10.2,-12.4)
{One real asymptotic direction,\\
no finite-offset affine asymptote};

\node[decision, fill=floworange, text width=28mm] (signfpp) at (6.2,-14.0)
{Sign of $f''(r)$?};

\node[proc, fill=floworangelight, text width=34mm] (doubleminus) at (3.4,-15.8)
{Two real branches
};

\node[proc, fill=floworangelight, text width=34mm] (doubleplus) at (9.0,-15.8)
{
no real branch};

\draw[arrow]
    (dzero.east)
    -- node[lab, above] {No}
    (slope.north |- dzero.east)
    -- (slope.north);

\draw[arrow] (slope.south) -- (mult.north);

\draw[arrow]
    (mult.west)
    -- node[lab, above] {simple}
    (simple.north |- mult.west)
    -- (simple.north);

\draw[arrow]
    (mult.south)
    -- node[lab, right] {double}
    (signfpp.north);

\draw[arrow]
    (mult.east)
    -- node[lab, above] {triple}
    (triple.north |- mult.east)
    -- (triple.north);

\draw[arrow]
    (signfpp.west)
    -- node[lab, above] {$<0$}
    (doubleminus.north |- signfpp.west)
    -- (doubleminus.north);

\draw[arrow]
    (signfpp.east)
    -- node[lab, above] {$>0$}
    (doubleplus.north |- signfpp.east)
    -- (doubleplus.north);

\end{tikzpicture}}
\caption{Flowchart for the linear-(linear-plus-cubic) classification. 
After choosing coordinates so that $P_1(\vb*{k})=k_z$ and $\ell_1(k_x,k_y)=k_x$, the problem reduces to the plane curve
$P_2(k_x,k_y)=k_x(1+a k_x^2+b k_xk_y+c k_y^2)+d k_y^3$.
The case $d=0$ is reducible and is classified by the quadratic factor $Q(k_x,k_y)+1$.
The case $d\neq0$ is irreducible and is classified by the multiplicity of the real roots of $f(t)=a+bt+ct^2+dt^3$.}
\label{fig:linear-linear-plus-cubic-flowchart}
\end{figure}
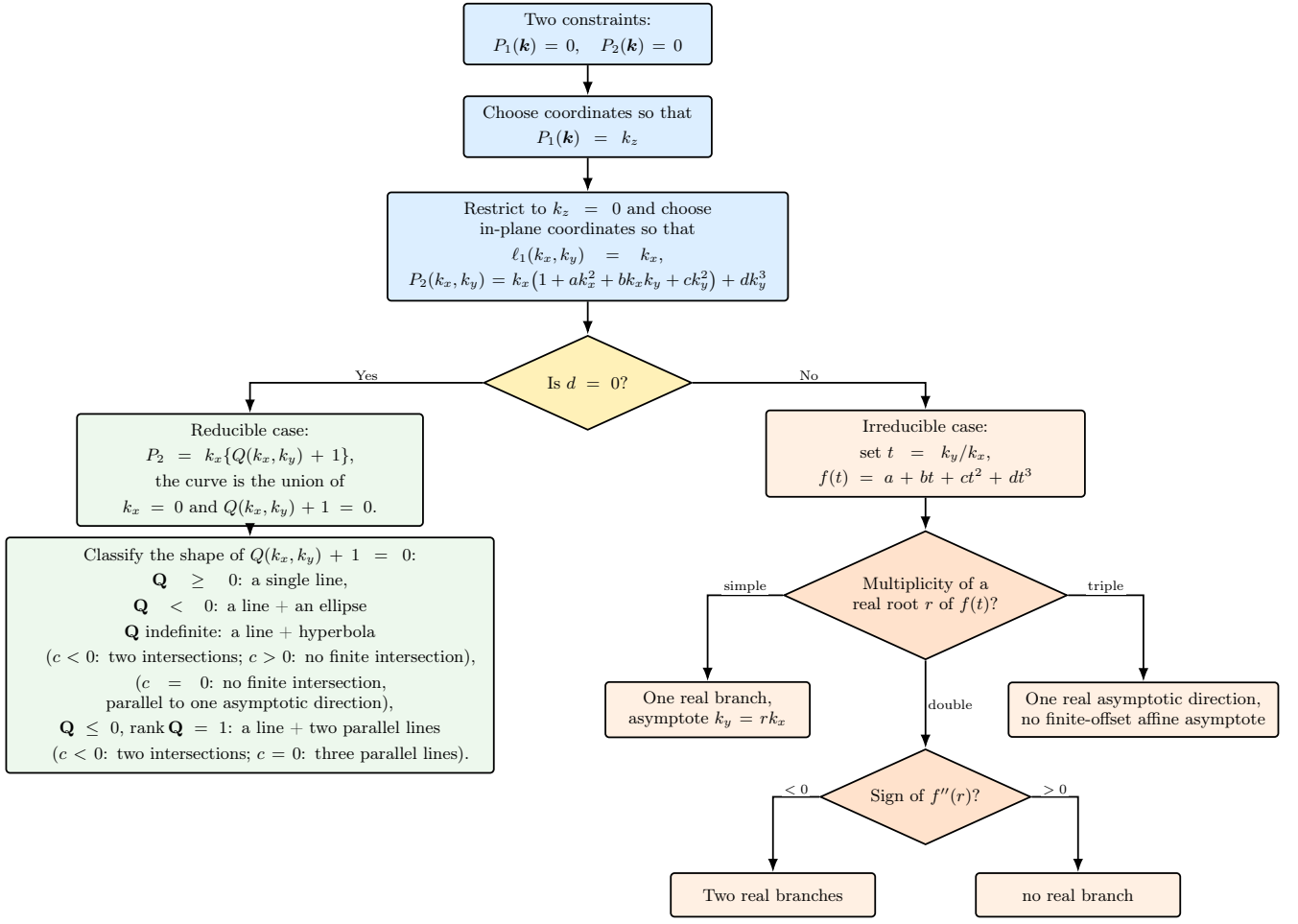

\subsection{The \(\ell_1=0\) branch}
The preceding linear-cubic classification assumes that the linear part of the second constraint remains nonzero after restricting to the plane defined by the first linear constraint. Equivalently, if the first constraint is $P_1=0$ and the linear part of the second constraint is $\ell$, the classification assumes
\begin{equation}
    \ell_1:=\ell|_{P_1=0}\neq 0 .
\end{equation}
This assumption is necessary in order to choose in-plane coordinates with \(\ell_1=k_x\) and to obtain the normal form displayed above.

There is, however, a separate degenerate case in which
\begin{equation}
    \ell_1=0 .
\end{equation}
Geometrically, this means that the two linear forms $P_1$ and $\ell$ are proportional. After restricting to the plane $P_1=0$, the second constraint has no linear part and reduces to a homogeneous binary cubic,
\begin{equation}
    P_2|_{P_1=0}=c_3(x,y),
\end{equation}
where $c_3$ is the restriction of the cubic homogeneous part to the plane $P_1=0$.

Since $c_3$ is homogeneous, its real zero set is a union of real lines through the origin. More explicitly, over $\mathbb R$ the binary cubic factors as
\begin{equation}
    c_3(x,y)=\prod_j L_j(x,y)^{m_j}\, Q(x,y),
\end{equation}
where the $L_j$ are distinct real linear factors, $m_j$ are their multiplicities, and $Q$ is either absent or a positive/negative definite quadratic factor with no real linear zero. Therefore the real set-theoretic zero locus is
\begin{equation}
    \{(x,y):c_3(x,y)=0\}
    =\bigcup_j \{(x,y):L_j(x,y)=0\} .
\end{equation}
Thus the common zero set of the original three-dimensional constraints is a union of straight lines through the origin inside the plane $P_1=0$.

Equivalently, the classification is determined by the real projective roots of the binary cubic $c_3$:
\begin{itemize}
    \item If $c_3$ has three distinct real projective roots, the real zero set consists of three distinct lines through the origin.
    \item If $c_3$ has one simple real projective root and one complex-conjugate pair, the real zero set consists of one line through the origin.
    \item If $c_3$ has repeated real roots, the real set-theoretic support is still the corresponding union of distinct real lines, while the multiplicities record tangent or higher-order contact data that are invisible in the set-theoretic support.
\end{itemize}

This $\ell_1=0$ branch does not produce a genuinely curved planar locus. It belongs to the homogeneous straight-line case rather than to the genuine linear-plus-cubic classification with $\ell_1\neq0$.
\section{Worm algorithm}

Here we describe the Monte Carlo update used for the lattice Hamiltonians in Eq.~\eqref{eq:lattice_rank_one_hamiltonian}. The construction applies to any two-component lattice Gauss's law of the form
\begin{equation}
    Q_{\vb*r}
    =(\Delta_1 S^z_1)_{\vb*r}
    +(\Delta_2 S^z_2)_{\vb*r},
    \qquad
    H=J\sum_{\vb*r}Q_{\vb*r}^2,
    \label{eq:SM_general_lattice_gauss_law}
\end{equation}
where \(S^z_{a,\vb*r}\) is identified with the electric-field component \(E_{a,\vb*r}\), and \(\Delta_a\) are translation-invariant finite-difference operators on the periodic lattice. Since the \(\Delta_a\) are polynomials in lattice translations, they commute with one another.

The elementary worm move is generated by a compactly supported scalar auxiliary field \(A_{\vb*r}\).  We propose a change of the longitudinal spin components
\begin{equation}
    \delta S^z_{1,\vb*r}=-(\Delta_2 A)_{\vb*r},
    \qquad
    \delta S^z_{2,\vb*r}=(\Delta_1 A)_{\vb*r}.
    \label{eq:SM_worm_null_update}
\end{equation}
This update lies in the null space of Gauss's law operator:
\begin{align}
    \delta Q_{\vb*r}
    &=
    \Delta_1\delta S^z_{1,\vb*r}
    +\Delta_2\delta S^z_{2,\vb*r} \nonumber\\
    &=
    -(\Delta_1\Delta_2 A)_{\vb*r}
    +(\Delta_2\Delta_1 A)_{\vb*r}
    =0 .
    \label{eq:SM_worm_constraint_preservation}
\end{align}
Thus the move exactly preserves \(Q_{\vb*r}\), and hence preserves the energy of the pure constraint Hamiltonian, up to floating-point roundoff.

In practice, one Monte Carlo attempt proceeds as follows.  A lattice center \(\vb*r_0\) is chosen uniformly, a compact local pattern \(A^{(0)}_{\vb*r-\vb*r_0}\) is chosen from a symmetric set of patterns, and an amplitude \(\alpha\) is drawn from a symmetric distribution, \(\alpha\in[-\alpha_{\max},\alpha_{\max}]\).  The auxiliary field is \(A_{\vb*r}=\alpha A^{(0)}_{\vb*r-\vb*r_0}\), and Eq.~\eqref{eq:SM_worm_null_update} gives the proposed new longitudinal components
\begin{equation}
    S^{z\prime}_{a,\vb*r}=S^z_{a,\vb*r}+\delta S^z_{a,\vb*r},
    \qquad a=1,2 .
    \label{eq:SM_worm_proposed_sz}
\end{equation}
If any proposed value violates the microscopic spin constraint, \(|S^{z\prime}_{a,\vb*r}|>1\), the whole move is rejected.  Otherwise it is accepted with Metropolis probability
\begin{equation}
    p_{\rm acc}
    =
    \min\{1,\exp[-\beta\Delta H]\},
    \qquad
    \Delta H
    =
    J\sum_{\vb*r}
    \left[(Q_{\vb*r}+\delta Q_{\vb*r})^2-Q_{\vb*r}^2\right].
    \label{eq:SM_worm_acceptance}
\end{equation}
For the exact null updates above, \(\Delta H=0\), so the only rejection in the \(T=0\) constrained sampler comes from the bound \(|S^z_{a,\vb*r}|\le1\).  A sweep consists of \(L^3\) such local attempts.

After an accepted update, the transverse components are refreshed so that each microscopic spin remains an \(O(3)\) unit vector.  For every updated site and sublattice, we draw an independent azimuthal angle
\begin{equation}
    \phi_{a,\vb*r}\in[0,2\pi)
\end{equation}
uniformly and set
\begin{equation}
    S^x_{a,\vb*r}
    =
    \sqrt{1-(S^z_{a,\vb*r})^2}\cos\phi_{a,\vb*r},
    \qquad
    S^y_{a,\vb*r}
    =
    \sqrt{1-(S^z_{a,\vb*r})^2}\sin\phi_{a,\vb*r}.
    \label{eq:SM_worm_transverse_refresh}
\end{equation}
This preserves the correct local \(O(3)\) measure, since \(d\Omega=d\phi\,dS^z\).
Equivalently, a single unconstrained \(O(3)\) spin has a flat measure in \(S^z=\cos\theta\) on the interval \([-1,1]\); the nonuniform distribution of \(S^z\) in the constrained ensemble is generated by the Gauss's law constraint and the hard bounds, not by an additional local Jacobian.  Therefore the symmetric additive proposal in \(S^z\), together with the Metropolis rule above, samples the desired measure, while the random azimuthal refresh samples the conditional transverse direction at fixed \(S^z\).

The proposal is symmetric: the inverse of the move generated by \(A_{\vb*r}\) is generated by \(-A_{\vb*r}\), and the center, pattern, and amplitude distributions are chosen independently of the current configuration.  Together with the Metropolis rule in Eq.~\eqref{eq:SM_worm_acceptance}, this gives detailed balance.  In the \(T=0\) runs, we initialize from \(S^z_{a,\vb*r}=0\), so that \(Q_{\vb*r}=0\), and the worm updates sample within the constrained ground-state manifold.  The measured structure-factor matrix is
\begin{equation}
    S_{ab}(\vb*{k})
    =
    \frac{1}{N}
    \left\langle
    S^z_a(\vb*{k})S^z_b(-\vb*{k})
    \right\rangle,
    \qquad
    S^z_a(\vb*{k})=\sum_{\vb*r}e^{-i\vb*{k}\cdot\vb*r}S^z_{a,\vb*r},
    \label{eq:SM_worm_structure_factor}
\end{equation}
For the symmetric channel used in Models II and III, we define
\begin{equation}
    C_\Sigma(\vb*{k})
    =
    S_{11}(\vb*{k})+S_{22}(\vb*{k})
    +2\,\mathrm{Re}\,S_{12}(\vb*{k}).
    \label{eq:SM_worm_symmetric_channel}
\end{equation}

\end{document}